\documentclass[10pt,peerreview,compsoc]{IEEEtran}

\usepackage{hyperref}
\usepackage{url}

\usepackage{tikz}
\usepackage{amssymb}

\usepackage[utf8]{inputenc}

\usepackage{color,soul}

\usepackage{multirow}
\usepackage{graphicx}

\usepackage{tabto}
\usepackage{amsmath}
\usepackage{pgfgantt}
\usepackage{tikz}
\usepackage{makecell}
\usepackage{algpseudocode}
\usepackage{mathtools}

\DeclarePairedDelimiter\floor{\lfloor}{\rfloor}
\usepackage{xspace} 
\usepackage{adjustbox}

\usepackage{ifthen}
\usepackage{xcolor}

\usepackage{booktabs} 

\usepackage{setspace}

\usepackage[linesnumbered,ruled,vlined]{algorithm2e}
\SetAlFnt{\scriptsize}
\usepackage{multicol} 
\usepackage[normalem]{ulem}
\useunder{\uline}{\ul}{}

\usepackage{lipsum,caption,graphicx}

\usepackage{amsthm}

\DeclareCaptionLabelFormat{andtable}{#1~#2  \&  \tablename~\thetable}

\newboolean{showcomments}
\setboolean{showcomments}{true}
\ifthenelse{\boolean{showcomments}}
{ \newcommand{\mynote}[3]{
    \fbox{\bfseries\sffamily\scriptsize#1}
    {\small$\blacktriangleright$\textsf{\emph{\color{#3}{#2}}}$\blacktriangleleft$}}}
{ \newcommand{\mynote}[3]{}}

\newcommand{\squeezeup}{\vspace{0mm}}

\newcommand*\numcircledmod[1]{\raisebox{.5pt}{\textcircled{\raisebox{-.9pt} {#1}}}}

\newcommand{\name}{NimbleChain\xspace}

\newcommand{\kblocks}{C blocks\xspace}

\newcommand{\bbcp}{BBP\xspace}

\newcommand{\chainselectionrule}{chain selection rule\xspace}
\newcommand{\CSR}{CSR}

\newcommand{\reddelivered}{committed\xspace}

\newtheorem{lemma}{Lemma}

\theoremstyle{remark}

\newtheoremstyle{theoremd}
  {}
  {}
  {\em}
  {}
  {\bf}
  {.}
  { }
  {\thmnote{ #3} \thmname{#1}}
\theoremstyle{theoremd}
\newtheorem{property}{property}

\title{\name: Speeding up cryptocurrencies in general-purpose permissionless blockchains} 

\author{Paulo Silva, Miguel Matos, Jo\~ao Barreto\\
INESC-ID, Instituto Superior T\'ecnico, Universidade de Lisboa\\
\{paulo.mendes.da.silva, miguel.marques.matos, joao.barreto\}@tecnico.ulisboa.pt}

\begin{document}

\maketitle

\begin{abstract}

  Nakamoto's seminal work gave rise to permissionless blockchains -- as well as a wide range of proposals to mitigate their performance shortcomings.
  Despite substantial throughput and energy efficiency achievements, most proposals only bring modest (or marginal) gains in transaction commit latency.
Consequently, commit latencies in today's permissionless blockchain landscape remain prohibitively high.

%
This paper proposes \name, a novel algorithm that extends permissionless blockchains based on Nakamoto consensus with a fast path that delivers 
\emph{causal promises of commitment}, or simply \emph{promises}.
Since promises only partially order transactions, their latency is only a small fraction of the 
totally-ordered commitment latency of Nakamoto consensus.
Still, the weak consistency guarantees of promises are \emph{strong enough} to correctly implement cryptocurrencies.
To the best of our knowledge, \name is the first system to bring together fast, partially-ordered transactions with consensus-based, totally-ordered 
 transactions in a permissionless setting.
This hybrid consistency model is able to speed up cryptocurrency transactions while still supporting smart contracts, which typically have (strong) sequential consistency needs.

We implement \name as an extension of Ethereum and evaluate it in a 500-node geo-distributed deployment.
The results show \name can promise a cryptocurrency transactions up to an order of magnitude faster than a vanilla Ethereum implementation, with marginal overheads.

\vspace{1cm}

\noindent\textbf{
If you’d like to cite this work, please use the reference below:\\\\
\emph{Paulo Silva, Miguel Matos, and João Barreto. \\
NimbleChain: Speeding up cryptocurrencies in general-purpose permissionless blockchains. \\
ACM Distributed Ledger Technologies. 2022.\\ 
\url{https://doi.org/10.1145/3573895}}
}

\end{abstract}

\IEEEpeerreviewmaketitle


\section{Introduction}
\label{sec:intro}


\begin{table*}[t]
  \center
\footnotesize
\begin{tabular}{|l|l|l|l|l|}
\hline
System                      & Event  & Double-spending-resistant & Consistency  & Latency  \\ \hline
Ethereum                    & Commit        & Yes                  & Total order  & $\sim$3min         \\ \hline
\multirow{2}{*}{\name} & Promise & Yes                  & Causal       & Typically $\sim$4s to $\sim$25s     \\ \cline{2-5} 
                            & Commit  & Yes                  & Total+causal & $\sim$3m         \\ \hline
\end{tabular}
	\caption{Comparison of \name's promise/commit events with Ethereum (an instance of BBP).}
\label{table:commits}
\end{table*}

The majority of permissionless blockchains, including mainstream ones
such as Bitcoin~\cite{Nakamoto2008} or Ethereum~\cite{Wood2014}, rely on the foundations
of Nakamoto's seminal
consensus protocol~\cite{Nakamoto2008}. 
In these systems, the probability that a given block 
has stabilized grows with the number of blocks that succeed it in the chain.
Hence, by setting a high enough threshold to consider a block -- and the transactions therein -- as committed,
one can ensure an arbitrarily low probability of the block being discarded.
This is commonly known as \emph{finality} and has been formalized as a \emph{persistence} property by Garay et al.~\cite{garay2015bitcoin}.

Permissionless blockchains can thereby be used as a (probabilistic) totally-ordered distributed ledger service 
with unique features, in particular their resilience to Sybil attacks which characterize permissionless environments.
Permissionless blockchains support a wide range of geo-distributed applications, from cryptocurrencies to
general-purpose smart contracts, but unfortunately they are also known for their poor performance.

In recent years, the research community has contributed with important improvements to permissionless blockchains -- such as to Nakamoto's longest chain rule~\cite{ghost, Wood2014, DBLP:conf/fc/LewenbergSZ15, cryptoeprint:2018:104, 254398}, hierarchical and parallel chains~\cite{Eyal2016, 10.1145/3087801.3087809, Yu2020ohie,10.1145/3319535.3363213}, sharded blockchains~\cite{10.1145/2976749.2978389, kokoris2018omniledger, 10.1145/3243734.3243853, EthSharding2018}, 
Layer-2 approaches~\cite{lightning, hao2018fastpay, mavroudis2020snappy, perun,tumblebit}
BFT-based blockchains~\cite{10.5555/3241094.3241117, DBLP:conf/wdag/PassS17, Abraham2017, Gilad2017, Miller2016, 10.1145/3341301.3359636}, and Proof-of-X alternatives \cite{Kiayias2017,zamfir2017casper,Gilad2017,Asayag2018, dziembowski2015proofs,BeyondHellman:2017, chen2017security}.
%
These approaches have focused mostly on improvements to throughput and/or energy efficiency, bringing only modest improvements to commit latency. 
Some notable exceptions reduce commit latencies by sacrificing security or scalability (as we discuss in \S\ref{sec:rw}).
In fact, as of today, commit latencies remain notably high in mainstream permissionless blockchains  --
around 1 hour in Bitcoin and 3 minutes in Ethereum.


Such high commit latencies are especially prohibitive for many applications
that require low-latency transactions.
Examples range from merchant applications
that need to deliver goods quickly ($<30$ seconds)~\cite{karame2012},
such as point-of-sale purchases and retail vending machines, take-away stores, online 
shopping, supermarket checkouts and bike sharing systems~\cite{dotan2020, karame2012, bamert2013have, guo2018bike}, to 
blockchain-backed IoT devices and applications~\cite{hao2018fastpay}.
To overcome the high commit delays, merchants and service providers frequently adopt risky 0- or 1-confirmation policies~\cite{dotan2020}, 
which accept transactions as granted well before the underlying blockchain can provide sufficiently strong guarantees of their persistence. Such policies
are inherently vulnerable to double-spending attacks. 

The commit latencies of mainstream permissionless blockchains reflect two well-known limitations of Nakamoto's consensus: for correctness,
blocks need to be generated (on average) at a slow pace with respect to network latency, and a transaction should only be considered as committed (i.e., persistent) after it is followed by a long sequence of (slowly generated) blocks in the chain \cite{garay2015bitcoin}.
To further complicate matters, a recent study~\cite{silva2020impact} concluded that, due to the emergence of powerful mining pools, 
blockchain systems should wait for even larger blockchain suffixes before committing, which means that commit latency is bound to increase, rather than decrease.

This paper focuses on improving the latency of general-purpose permissionless blockchains. 
Our insight is based on the observation that the actual consistency needs of cryptocurrencies, which
constitute the bulk of today's most important permissionless blockchains
(i.e., almost 100\% in Bitcoin~\cite{10.3389/fbloc.2019.00007} and 44\% in Ethereum~\cite{silva2020impact}), 
can be satisfied without resorting to consensus~\cite{guerraoui2019consensus}.
With this in mind, we set out to build a general-purpose permissionless blockchain that offers a 
hybrid consistency model where two distinct transaction handling paths gracefully coexist, each one serving applications with differing consistency
needs. In detail, a novel fast path that handles weakly-consistent cryptocurrency transactions, 
together with the existing slow path that processes the remaining strongly-consistent transactions -- such as smart contract transactions -- in a consensus-based, 
totally-ordered fashion. 

This paper comprises contributions to achieve the above vision. \textbf{As a first contribution,}
we propose an extension to the traditional issue/commit model, by introducing the notion of \emph{causal promise of commitment}, or simply \emph{promise}. 
A promise is a new event in a transaction’s life cycle.
Informally, when a process promises a transaction $t$, that means that $t$ will be eventually committed by every correct process, even if $t$ is part of a double-spending attempt by a malicious user, while satisfying any causal dependencies of $t$.
Promises offer weaker consistency guarantees than traditional blockchain consensus -- most importantly, promises are only partially ordered.
Therefore, promises can be implemented by a faster protocol in a small fraction of the commitment latency.
Despite weaker, the guarantees of the promise event are still \emph{strong enough} to fulfil the consistency needs of cryptocurrencies. 
Namely, they are causally ordered and resist double-spending attacks. 
Table~\ref{table:commits} summarizes the key differences between promising and committing a transaction. 


\textbf{As a second contribution,}
we propose \name, a novel permissionless blockchain system that
extends permissionless blockchains based on the Bitcoin Backbone Protocol (BBP)~\cite{garay2015bitcoin} with a \emph{partially ordered promise fast path}.
This fast path allows cryptocurrency transactions, which
constitute the bulk of today's most important permissionless blockchains
(i.e., almost 100\% in Bitcoin~\cite{10.3389/fbloc.2019.00007} and 44\% in Ethereum~\cite{silva2020impact}) to be promised substantially faster than the original commit, in a large portion of correct processes.
Since today's general-purpose blockchains also support smart contracts, which typically have sequential consistency needs,
\name also supports strongly consistent transactions. Together, offers a \emph{hybrid} consistency model that is able to handle transactions  supports both types of transactions, as Table~\ref{table:commits} outlines.
To the best of our knowledge, \name is the first permissionless blockchain system to bring together fast partially ordered transactions with totally ordered, consensus-based transactions in a permissionless setting. 

\textbf{As a third contribution, } we implement \name as an extension of the Ethereum blockchain, which demonstrates the practicality of our proposal. 

\textbf{As a final contribution, } an evaluation with 500 processes 
in a realistic geo-distributed environment that shows that \name's
  fast path reduces the commit latencies of cryptocurrency transactions by an order of magnitude with negligible overhead when compared with Ethereum.

The rest of the paper is organized as follows.
\S\ref{sec:background} provides background on permissionless blockchains and describes a generic baseline protocol. 
\S\ref{sec:promise} defines the notion of promises and \S\ref{sec:parmenides} shows how \name extends the baseline with promises.
\S\ref{sec:cryptopromises} describes how we can leverage \name’s promises to implement low-latency cryptocurrencies.
\S\ref{sec:evaluation} evaluates our implementation of \name as an extension of Ethereum in a large-scale geo-distributed scenario.
\S\ref{sec:discussion} discusses the limitations of \name.
\S\ref{sec:rw} surveys related work, and \S\ref{sec:conclusion} concludes the paper.

\begin{table}[]
  \center
\footnotesize
\begin{tabular}{|l|l|}
\hline
Symbol 	& Meaning\\ \hline
$D$		& Maximum message delivery delay\\ \hline
	$BBP$	& Bitcoin Backbone Protocol~\cite{garay2015bitcoin} \\ \hline
$B$		& Average block generation time\\ \hline
$C$		& BBP commit threshold (minimum chain suffix height)\\ \hline
\end{tabular}
\caption{Notation used in the paper}
\label{table:symbols}
\end{table}

\section{Background}
\label{sec:background}



Today's most valuable permissionless blockchains, such as Bitcoin or Ethereum, are based on the foundations laid by the consensus protocol proposed by Nakamoto~\cite{Nakamoto2008}.
Different studies~\cite{garay2015bitcoin,pass2017analysis,10.1145/3243734.3243814,10.1007/978-3-319-63688-7_10} have formally analyzed the common core of such blockchain protocols. 
Hereinafter, we adopt the term coined by Garay et al.~\cite{garay2015bitcoin}, \emph{Bitcoin Backbone Protocol} (BBP), to generically refer to the core protocol behind today's permissionless blockchains inspired by Nakamoto's consensus protocol.

Garay et al. ~\cite{garay2015bitcoin} established two main properties of BBP, persistence and liveness (further discussed below), 
and formally studied them under a strict (and unrealistic) network model. Subsequent works showed such properties also hold for more realistic network models~\cite{pass2017analysis,10.1145/3243734.3243814,10.1007/978-3-319-63688-7_10}.
In this work, we consider BBP under the latter models as our starting point, which we extend to build \name. 

In the next sections, we present the system model, the main properties of BBP, and conclude with an overview of the protocol.





\subsection{Assumptions}
\label{sec:bgsysmodel}
The BBP runs on a peer-to-peer network of processes and relies on the following assumptions~\cite{pass2017analysis,pass2017rethinking,Kiayias2017,10.1145/3243734.3243814}.
First, Byzantine adversaries control less than 50\% of the total mining power that is used to produce blocks 
and do
not have computational power to subvert the standard cryptographic primitives.
Any message broadcast by a correct process is delivered to every correct process with high probability, and with a maximum delay of $D$~\cite{pass2017rethinking,pass2017analysis,rocket2020scalable,Kiayias2017}.
This is often a hidden assumption of permissionless blockchains systems but it is in fact a critical aspect of the system that we do not only expose but embrace.
In fact, in Nakamoto consensus, the difficulty of the Proof-of-Work (PoW) puzzle is based on the maximum delivery delay D~\cite{pass2017analysis,rocket2020scalable}, and hence this is an assumption shared by Bitcoin and Ethereum.
The same applies to proof-of-stake blockchains such as Ouroboros~\cite{Kiayias2017} and Algorand~\cite{Gilad2017}, which assume a bounded delay. 
The correctness of BBP, and systems such as Bitcoin or Etherum, thus depends on the premise that the system takes much longer to generate a new block than to propagate it~\cite{garay2015bitcoin}.
Let us denote $B$ as the overall average block generation time.
Hence, we assume that $D/B$ is relatively small~\cite{garay2015bitcoin,pass2017analysis,cryptoeprint:2015/1019,10.1145/3243734.3243814}.

For presentation simplicity, we assume that all processes can produce (\emph{i.e.} mine) blocks and act as clients by submitting transactions to the system.
Moreover, we also assume that each correct client is collocated with a correct process participating in the protocol. Devising robust solutions to support correct clients that interact with possibly Byzantine processes running the protocol is an open problem, which is orthogonal to our work~\cite{10.1145/3422337.3447832}. 

A transaction is signed by the process that issues it.
The body of a transaction is an application-dependent payload (e.g. the target account and amount to transfer in a cryptocurrency transaction,
or, in a smart contract transaction, the target smart contract, its method and arguments).

Each transaction also carries a local sequence number, which is consecutively unique for that issuer\footnote{This is the case, for instance, with Ethereum's transaction nonces. Still, we note that BBP can generically abstract blockchain protocols without per-transaction sequence numbers.}.
Note that a Byzantine process $p$ can submit two or more distinct transactions with the same sequence number, with the intention of having some correct processes commit one transaction (and, hence, discard the other one(s) as invalid), and other corrects processes decide in the opposite direction. 
For any pair of distinct transactions that share the same sequence number, 
we say that each transaction is a \textit{double spending} of the other (and vice-versa). 

We summarize the notation used in the paper in Table \ref{table:symbols}.

\subsection{Properties}
\label{sec:bgproperties}

We start by defining the event of committing a transaction as follows.
If a correct process $p$ has a block $b$ in its blockchain followed by at least $C$ other blocks, we say that every transaction included in $b$ is \emph{committed} at $p$.


Garay et al.~\cite{garay2015bitcoin} define the following two main properties of the commit event of BBP, with high probability:%
\footnote{For presentation simplicity, we present a formulation that is simpler than Garay et al.’s original one. Namely, we omit some parameters that are orthogonal to our contributions in this paper, and explicitly use the term \emph{commit}.}





\begin{property}[Persistence]
If a correct process $p$ commits a block $b$, and consequently the transactions in $b$, then any correct process has block $b$ in the same position in the blockchain, from this moment on. 
\end{property}

\begin{property}[Liveness] 
If a correct process submits a transaction $t$ then all correct {}process eventually commit $t$.
\end{property}

%
%
%
%

Different works have formally studied the above properties, using different methods. Starting from Garay et al.'s initial analysis, which considered simplified network assumptions \cite{garay2015bitcoin}, subsequent works \cite{pass2017analysis,10.1145/3243734.3243814,10.1007/978-3-319-63688-7_10} have formally shown that BBP guarantees such properties under the more realistic assumptions that we consider in \S\ref{sec:bgsysmodel}. 

Two important corollaries can be easily drawn from the above properties.
The first one is that, if a correct process $p$ commits a transaction $t$, then all other correct processes will eventually commit $t$.
A second corollary is that, if a correct process $p$ commits transaction $t_A$ and, later, transaction $t_B$, then any other correct process $q$ will
commit both transactions in the same order -- in other words, the commit event is totally ordered.

Together, the above properties and the associated corollaries
constitute strong consistency guarantees, which BBP provides with high probability. 
These enable many geo-distributed applications
to operate even when operating under an adversarial permissionless environment.
For instance, 
state-machine replication can be built on top of \bbcp through smart contracts. 
Yet, these guarantees are provided at the cost of a high latency.

\subsection{\bbcp Algorithm} 
At a high-level, the algorithm works as follows.
Each process holds a local ledger of transactions that consolidates two components: a local copy of the blockchain (or, simply, local chain) and the \emph{mempool}. 
The local chain is organized as a cryptographically linked list of blocks.
Blocks have a monotonically increasing gap-free sequence number and each block includes a totally-ordered sequence of transactions. 
The mempool is a local queue of individual transactions not yet included in the local chain.
The transactions in the mempool are ordered after the transactions in the local chain.
When receiving a transaction, a process  
performs a series of  validations 
such as checking for funds, checking whether another transaction 
with the same identifier already exists in the local ledger (\emph{i.e.} a double-spend attempt) and 
verifying the transaction's digital signature, among others.
If the transaction is valid and not yet in the local chain, it is inserted in the mempool.


Processes initially share a genesis block $b_0$ and produce blocks by selecting a subset of transactions in the mempool and creating a proof that depends on these transactions and on the last block in the local chain. 
The proof is a Sybil-proof leader election mechanism, such as PoW, that ensures the process legally produced the block.
The puzzle difficulty of PoW is a function of $D$~\cite{pass2017rethinking}.


Upon producing a new block, a process appends it to its local chain
and broadcasts it.
When a block is delivered, it is validated before being included in the local chain.
Due to concurrency, 
it is possible that two or more processes produce 
two competing blocks $b_{i+1}$ and $b'_{i+1}$ 
that extend the sequence $(b_0, \ldots,b_i)$.
Processes will select either sequence $(b_0, \ldots,b_i, b_{i+1})$ or 
 $(b_0, \ldots,b_i, b'_{i+1})$,
using a \chainselectionrule (\CSR).
In Nakamoto consensus, the \CSR\xspace corresponds to selecting the chain with greatest PoW effort.
By relying on the \CSR, the protocol will make the system converge
by gradually agreeing on an increasing \textit{common blockchain prefix}.
%
Blocks, and the transactions therein, are considered \reddelivered, when
followed by a sequence of \kblocks.

\subsection{The weak consistency needs of cryptocurrencies} 
\label{sec:consensuslesstx}

Our work hinges on the observation that cryptocurrency transactions, also known as asset transfer transactions, do not need to be totally order~\cite{guerraoui2019consensus}. 
The intuition behind this observation is the following.
In an asset transfer scenario, the system maintains a set of accounts where each account maintains the number of assets held, for instance the amount of cryptocurrency in a given account. 
Each account has a single owner that can withdraw funds from this account and transfer them to other accounts. The other users in the system can only read the accounts balance and transfer funds to it (from their own accounts). 
Therefore, it is the sole responsibility of an account owner to order the assets withdrawn from that account.

Intuitively, this removes the need to order a transaction that withdraws assets from a given account relatively to every other transaction 
that is concurrently withdrawing from distinct accounts.
Therefore, instead of a total order across asset transfer transactions, Guerraoui et al.~\cite{guerraoui2019consensus} 
have shown that the semantics of asset transfer only require causal order and some mechanism that prevents double-spending attempts.

This observation unveils a notable opportunity to implement faster cryptocurrencies
that escape the high cost of agreeing on a total order.
This opportunity has been explored by recent proposals for weakly-consistent cryptocurrency protocols for permissionless systems, not just by Guerraoui et al. \cite{guerraoui2019consensus,collins2020online} but also by others \cite{sompolinsky2016spectre,otte2017trustchain,sliwinski2019abc,rocket2020scalable,DBLP:conf/wdag/KuznetsovPPT21}(more details in \S\ref{sec:rw}).

However, such proposals are not compatible with mainstream permissionless blockchains, whose general-purpose nature 
requires also serving transactions issued by applications with strong, sequential consistency requirements 
-- such as those issued to/by smart contracts. 
To the best of our knowledge, bringing weakly consistent cryptocurrency transactions to mainstream permissionless blockchains 
that also support strongly consistent smart contracts remains an open problem.
We address such problem in the next sections. 

\section{Promises}
\label{sec:promise}

\name extends the life-cycle of a transaction, $t$, with an additional event, a \emph{promise of commitment of $t$}, 
which we abbreviate to \emph{promise of $t$}.
Whenever a process $p$ is able to anticipate -- through some means, as we discuss later on --
that $t$ will eventually commit (by BBP), $p$ can immediately promise $t$.
The following property formally captures this notion.

\begin{property}[Eventually committed upon promised] 
If a correct process promises a transaction $t$ then all correct process eventually commit $t$.
\end{property}

\begin{figure*}[t]
\includegraphics[width=\textwidth]{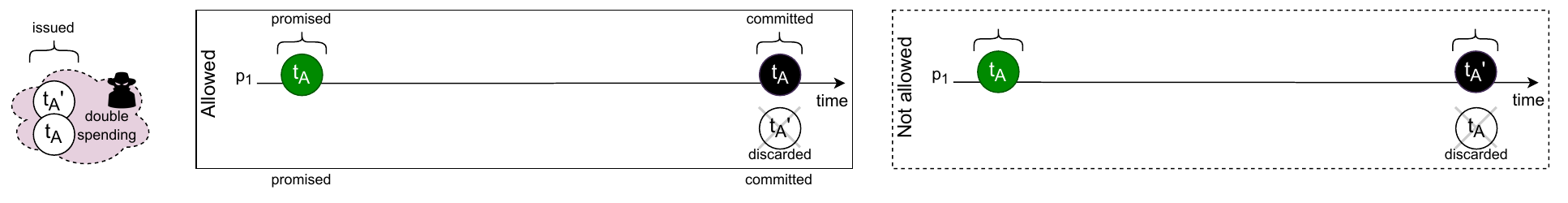}
\caption{An example illustrating the \emph{eventually committed upon promised} property. 
Transaction $t_A$, which belongs to a double-spending attack (along with $t_A'$) is promised by correct process $p_1$.
By the above-mentioned property, every correct processes will eventually commit $t_A$ and, consequently, discard its double-spending counterpart, $t_A'$.}
  \label{fig:promise-eventual-commit}
  \squeezeup
\end{figure*}

Figure \ref{fig:promise-eventual-commit} illustrates this property with an example.
Process $p_1$ promises transaction $t_A$, thus every correct process (including $p_1$) will eventually commit $t_A$.

This example also illustrates that promises resist double-spend attacks, similarly to commits.
In our example, a malicious user issues two conflicting transactions, $t_A$ and $t_A'$
in an attempt to double-spend. 
Recall that BBP ensures that, if one of such transactions commits, then the other conflicting
transaction(s) will not, with high probability.
Hence, since correct process $p_1$ promises $t_A$, then 
that process can immediately infer that, even if a double-spending attempt threathens $t_A$, 
the system will eventually commit $t_A$ and, thus, discard any conflicting transactions
($t_A'$ in the example).
This implies that an alternative example where the system instead committed $t_A'$ (Figure \ref{fig:promise-eventual-commit} right) is 
not allowed.

To be useful, promises must be delivered faster than commits.
To accomplish that,
we hinge on the observation that an important class of applications
do not require the strong consistency guarantees of commits.
Therefore, promises offer weaker, but still \emph{strong enough},
guarantees for such applications.
As we show in \S\ref{sec:cryptopromises}, cryptocurrencies are one notable example
that can be implemented by the weak guarantees of promises.


\begin{figure*}[t]
\includegraphics[width=\textwidth]{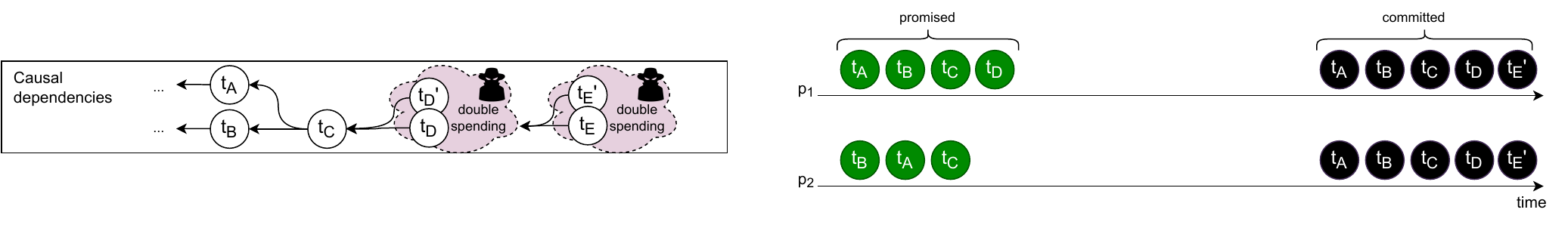}
\caption{An example in a 2-process system illustrating how promises are weaker than commits. Only promise and commit events are shown (we omit issue events). 
The example illustrates the two ways in which the promise order is more relaxed than the commit order: i) promise order is partial (more precisely, a causal order), not total; 
and ii) not every transaction that commits needs to have been promised at every process.}
  \label{fig:promise-simple}
  \squeezeup
\end{figure*}

Compared to commits, promises are fundamentally weaker
along two dimensions, which Figure \ref{fig:promise-simple} illustrates with an example.
As a first relaxation, promised transactions 
are only partially ordered by causal dependencies
among them. In other words, promises occur in causal order.

For instance, each process in Figure \ref{fig:promise-simple} promises transactions in different orders, 
which all satisfy the causal dependencies between transactions.
In contrast, BBP commits the transactions in total order, 
as depicted on right-hand side of Figure \ref{fig:promise-simple}.
Causally ordering events in a distributed system
is well-studied to be easier than totally ordering them \cite{causal95},
hence implementing promises is an easier problem than the one solved by BBP.


As a second relaxation, correct processes do not need to agree on promising a given transaction.
A transaction $t$ may be promised at some correct process, but not necessarily at all other correct process.
This relaxation allows best-effort implementations that, under worst-case scenarios,
can simply give up their efforts to promise some transaction at a subset of processes.

For example, transaction $t_D$ in Figure \ref{fig:promise-simple} is promised at process $p_1$ 
but not at $p_2$. 
Although $p_2$ does not promise $t_D$, $p_2$ later commits $t_D$.
Since commits ensure strictly stronger properties than promises,
the commit of $t_D$ properly replaces the missing promise.
Concretely, for a client at $p_2$ with weak consistency demands that wishes to know
when $t_D$ has been promised, the fallback is to wait longer
until $t_D$ commits (or, instead, is discarded by BBP).

At an extreme scenario, transactions $t_E$ and $t_E'$ (a double-spending pair) are not promised 
at any process. In this case, processes will eventually learn that one of such transactions 
has committed ($t_E'$ in our example). 
This scenario is different than the previous ones because, since no process promises neither $t_E$ nor $t_E'$, 
the \emph{eventually committed upon promised} property does not hold and, thus, a correct system is free to 
agree on either transaction to be committed at every correct process -- concretely, an 
alternative example where every correct process committed $t_E$ would also be correct.

On a related note, suppose that, after $t_E'$ commits at $p_2$, that process receives a new transaction,
$t_F$, that causally depends on $t_E'$.
By the same rationale as above, we allow $p_2$ to promise $t_F$ since its causal
dependency has already committed (instead of promised) at $p_2$.
%
%
Leveraging the above observations, we can finally formulate the (partial) causal order of promises, 
as follows.

\begin{property}[Causal promise order]
No correct process promises a transaction $t$
before promising or committing $t$'s causal dependencies.
\end{property}

Hereafter, we assume that transactions carry a metadata field that denotes its 
causal dependencies (using some suitable causality tracking mechanism).
Also, a transaction $t$ issued by correct process $p$ causally depends on every transaction
previously issued by $p$\footnote{This is a usual definition in distributed algorithms with causal order \cite{lamport78,causal95}.}.

As a final requirement, which has been implicit so far, we must enforce that BBP's commit order also
respects the causal dependencies that determine the causal promise order.
This ensures the commit total order is a linear extension of the causal
promise orders observed at the different processes.
This requirement prevents that the commit order introduces causality anomalies 
that are absent in the promise order.

\begin{property}[Causal commit order] 
No correct process commits a transaction $t$
before committing $t$'s causal dependencies.
\end{property}

Finally, we remark that, in general-purpose blockchains, different applications or operations within an application, might have
distinct consistency demands. This explains why we decided to add the promise event to the original issue/commit model, rather than replacing the commit event with the promise event. The richer issue/promise/commit event set offers a hybrid consistency model. It allows different applications, or different operations from the same application, to choose whether to wait for a transaction to be promised or committed, depending on the required guarantees.

\section{\name}
\label{sec:parmenides}

The main goal of \name is to provide a fast path that promises transactions
much sooner than BBP commits them.
To accomplish that goal, we exploit the fact that the \emph{partial} causal order
across promised transactions
is fundamentally easier to achieve than the
\emph{total} causal order that is required across committed transactions.

Besides the promise fast path, \name still relies on BBP, for two main reasons.
First, \name's fast path is best-effort, since some
unlikely scenarios caused by double-spending attempts prevent it from
promising specific transactions.
To handle such cases, \name uses BBP to let
clients determine whether an unpromised transaction is effectively committed or not.
The BBP slow path also serves to provide a total causal order
when application semantics have strong consistency demands.
Table \ref{table:commits} summarizes the essential differences between both paths.

Next, we present \name's design. 
\S\ref{sec:algoverview} starts with an overview of \name. 
\S\ref{sec:mainalg} then details the main algorithm of \name. 
\S\ref{sec:slicing} discusses how \name can be tuned to maximize its resilience and \S\ref{sec:causal} details how \name enforces causal order.

\subsection{Overview}
\label{sec:algoverview}

\begin{figure*}[t]
\center
\includegraphics[width=0.7\textwidth]{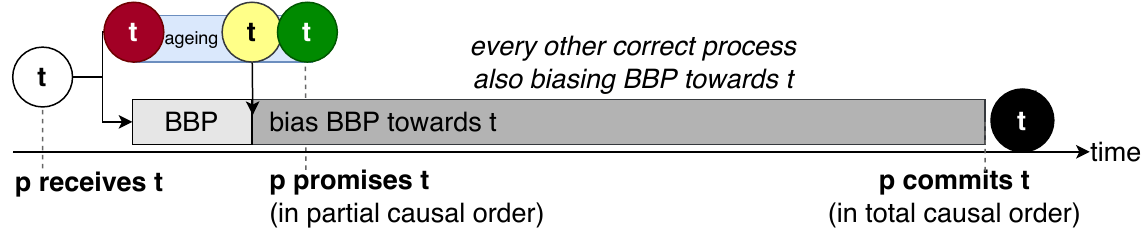}
\caption{An example where \name's ageing fast path promises a transaction, $t$. Later, the biased BBP slow path
  eventually commits $t$, even if $t$ was issued by an attacker that issued a double-spending transaction that is competing with $t$.}
  \label{fig:overview-promise}
  \squeezeup
\end{figure*}

In \name, any transaction that is received by a process (either issued locally or received from another issuer)
is fed simultaneously into two modules, which handle the transaction in parallel.
We first present the high-level semantics of each module (we defer
their algorithms to the following section).
Next, we describe how both building blocks are used together to
support fast transaction promises.
For presentation simplicity, we here assume that transactions have no
causal dependencies. 
We lift this simplification in \S\ref{sec:causal}.

\subsubsection{Fast path}
The fast path runs a transaction \emph{ageing} algorithm.
As a transaction $t$ ages at some process $p$, its current age is labelled as
a color (red, yellow or green), which we denote by $color_p(t)$.
A transaction is initially red (when received at the process), but may later transition to yellow, and finally to green,
as Figure \ref{fig:overview-promise} illustrates.
  
Not all transactions reach the green state; in fact, double-spending attempts
may cause some transactions to stop their aging in red or yellow (see \S\ref{sec:mainalg}).
In any case, it takes a short time (considerably shorter than BBP's latency)
for any process to determine the final age of a transaction.

The ageing is determined locally at each process, with no cross-process coordination
taking place. A consequence is that it cannot guarantee that
every correct process reaches the same final age for a given transaction.
Instead, it ensures a weaker guarantee:

\begin{property}[Bounded age consistency]
If at least one process, $p_i$, ages some transaction, $t$, up to green ($color_{p_i}(t)=green$),
then any other correct process $p_j$ must have aged $t$ up to yellow or green ($color_{p_j}\in\{yellow, green\}$).
  \end{property}

\subsubsection{Slow path}
The second module consists of BBP augmented with a novel feature, called \emph{biased BBP commit},
which we will describe shortly.
Upon receiving some transaction $t$ as input, the BBP's consensus protocol eventually
decides whether $t$ commits within the total order of the blockchain, or not.
Since this takes a long time, BBP constitutes the slow path of \name.

For any transaction $t$, each process has the option to \emph{bias BBP towards
$t$}.
By default, this option is disabled for new transactions that are received by \name.
If activated by every correct process for a common transaction, the biased BBP option
provides a powerful guarantee:

\begin{property}[Biased transaction selection]
If every correct process biases BBP towards a given transaction $t$, then
  BBP will eventually commit $t$, even if double-spending transactions
(of $t$) have also been issued.
\end{property}

\subsubsection{Putting it all together}
\name uses the ageing protocol as a best-effort mechanism to trigger promises.
When a process $p$ ages some transaction $t$ up to green, $p$ can immediately promise $t$. 

Of course, \name needs to ensure that, if at least one correct process promises $t$,
then BBP eventually commits $t$ at every correct process.
This is not trivial since $t$ may have been issued by an attacker that,
concurrently, also issued a double-spending transaction, $t'$; hence, if no additional care was taken,
BBP might end up
committing $t'$ and, thus, discarding $t$.
To address this, \name connects the ageing protocol and the biased BBP feature as follows:
as soon as a process $p$ has aged a transaction up to (at least) yellow,
$p$ biases BBP towards $t$. This is depicted in Figure \ref{fig:overview-promise}.

From the above-mentioned properties (bounded age consistency and
biased transaction selection),
one can conclude that this approach ensures that, if at least one process $p$
promises a transaction $t$ in the fast path,
then $t$ will be eventually committed (by BBP) at
every correct process.
In fact, if $p$ promises $t$ (hence, $t$ has aged up to green in $p$),
then $t$ will at least age up to yellow at every other correct process.
Therefore, every correct process will bias BBP towards $t$. Consequently,
BBP will eventually commit $t$ -- even in the presence of a
double-spending transactions (conficting with $t$).

It is worth noting that there are scenarios, caused by double-spending attempts, where a process $p$ cannot promise
a transaction $t$ and, therefore, needs to wait for BBP's
slow path to learn whether $t$'s outcome. 
We discuss such scenarios in the next section. 

\subsection{Main algorithm of \name}
\label{sec:mainalg}

\begin{figure*}[t]
\includegraphics[width=\textwidth]{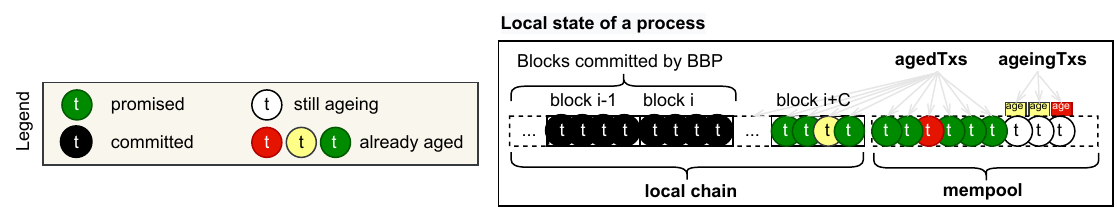}
  \caption{Local state of a process in \name.}
  \label{fig:state}
  \squeezeup
\end{figure*}

\name comprises two main modules,
the transaction ageing and the biased BBP commit.
Both run on top of an instance of BBP, intercepting specific routines of BBP.
We take advantage of the extensible design of BBP, customizing it with relatively simple extensions
that affect the incoming transaction and chain validation routines%
\footnote{More precisely, inside the \emph{content validation predicate ($V(\cdot)$)} and the \emph{environment} (more precisely, where the \emph{input tapes} of each process are determined) according to Garay et al.’s terminology \cite{garay2015bitcoin}, which are external to BBP and, thus, can be modified without hurting the correctness of the underlying BBP.}.
We provide details on our concrete extension of Ethereum in \S\ref{sec:implementation}.

At each process, both modules share the state 
depicted in Figure~\ref{fig:state}.
This includes BBP's state; namely, a local chain and a mempool.
Moreover, each process also maintains each transaction's age (or color),
stored in the \emph{agedTxs} and \emph{ageingTxs} sets which comprises transactions that have already reached their final age, and transactions whose ageing is still ongoing, respectively.

We now present the algorithms of each module.

\begin{algorithm}[t]
  \caption{Ageing protocol (\name’s fast path)}\label{alg:ageing}

  \SetKwProg{Initially}{Initially}{}{}

 \SetKwData{ageing}{ageingTxs}
  \SetKwData{aged}{agedTxs}

\Initially{}{
        $\ageing \gets \varnothing$\;
        $\aged \gets \varnothing$\;
}

\SetKwProg{Fn}{Function}{:}{}

  \SetKwFunction{FAge}{age}
  \Fn{\FAge{transaction $t$)}}{
        $(t_{aged}, ageval) \gets lookup(\aged, t.issuer, t.seqno)$\;
        \If{$t_{aged}$}
        {
        \KwRet{$ageval$}\;\nllabel{alg:ageing:age:aged}
        }
        $(t_{ageing}, ts) \gets lookup(\ageing, t.issuer, t.seqno)$\;
        \eIf{$t_{ageing}$}
        {
        \KwRet{$(now - ts)/D$}\;\nllabel{alg:ageing:age:ageing}
        }
        {
        \KwRet {$\perp$}\;\nllabel{alg:ageing:age:none}
        }
  }
  
  \SetKwFunction{FPromise}{promise}
  \SetKwProg{Thread}{Procedure}{:}{}
  \Thread{AgeingMonitor}{
	\While{true}
	{
		\For{each $t$ in \ageing}
		{
			\If{$age(t) = AT$ \nllabel{alg:ageing:AT}}
			{
				remove(\ageing, $t$)\;
				add(\aged, ($t$, $AT$))\;
				\FPromise{$t$}\;\nllabel{alg:ageing:promise}
		}
	}
  }
}
	
    \SetKwProg{NewTx}{Upon receiving}{:}{}
    \NewTx{transaction $t$} {
    	$t_{prev} \gets lookup(\aged \cup \ageing, t.issuer, t.seqno)$\;
	\eIf{$t_{prev}=null$} {
		add(\ageing, $t$, now)\;\nllabel{alg:ageing:startageing}
	}
	{
		\If{$t_{prev} \neq t$}
		{
			reject $t$ from mempool\;\nllabel{alg:ageing:reject}
			\If{$isAgeing(t_prev)$} 
			{			
				add(\aged, ($t_{prev}$, $age(t_prev)$))\;
        remove(\ageing, $t_{prev}$)\;\nllabel{alg:ageing:stopageing}
			}
		}
	}

}

  \SetKwFunction{FColor}{color}
  \Fn{\FColor{transaction $t$)}}{\nllabel{alg:ageing:color}
    \Switch{$age(t)$}{
            \Case{$\perp$}{
              \KwRet{$\perp$}
            }
            \Case{$AT$}{
              \KwRet{green}
            }
            \Case{$AT-2$}{
              \KwRet{yellow}
            }
            \lOther{
              \KwRet{red}
            }
    }
  }

\end{algorithm}

\subsubsection{Fast path: transaction ageing}

The transaction ageing fast path runs Algorithm \ref{alg:ageing}.
A transaction $t$ can be ageing, be already aged,  or in none of such states.
The age of a transaction $t$ at a given process is an integer value,
determined
by function \emph{age} in Algorithm \ref{alg:ageing}.
This value is translated to a color, as we explain later on.
Before a transaction starts to age, we simply say that its age/color is undefined (denoted by $\perp$). 

When receiving a transaction $t$, a process $p$ assigns a local physical timestamp to $t$
and adds $t$ to \emph{ageingTxs} (ln. \ref{alg:ageing:startageing}, Alg. \ref{alg:ageing}).
From that instant on, we say that $t$ is \emph{ageing}.
While a transaction $t$ is ageing, its age
is given by how much time has already passed since $t$'s timestamp, measured in units of $D$ (ln. \ref{alg:ageing:age:ageing}). 
The ageing of $t$ stops in two situations: 

\begin{enumerate}

\item When the age of $t$ reaches a
system-wide \emph{ageing threshold} (AT) (ln. \ref{alg:ageing:AT}).
For presentation simplicity, in this section we
assume $AT=4$. 
In \S\ref{sec:slicing}, we show that lower $AT$ values are not appropriate 
and, more importantly, higher values can achieve higher resiliency under targeted attacks at the expense of latency.

\item When the local process 
receives a double-spend of $t$ (ln. \ref{alg:ageing:stopageing}).

\end{enumerate}

In either case, we say that $t$ has aged, and its age is final.
In the case of double-spending attempts, where at least two conflicting transactions $t$ and $t'$ are 
received by a process $p$, only the first one of them that reaches $p$ is accepted into the mempool and
ages at $p$. 
The remaining (double-spending) transactions are rejected from the mempool (ln. \ref{alg:ageing:reject}) and,
consequently, will
always have an undefined age ($\perp$) at $p$.

The age values are translated into a color (ln. \ref{alg:ageing:color}).
When a transaction $t$ starts ageing, its color is red.
As soon as its age reaches $AT-2$, its color changes to yellow.
Finally, when its age is $AT$, $t$'s color changes to green and,
consequently, the local process promises $t$ (ln. \ref{alg:ageing:promise}).

\begin{figure*}[t]
\center
\includegraphics[width=0.5\textwidth]{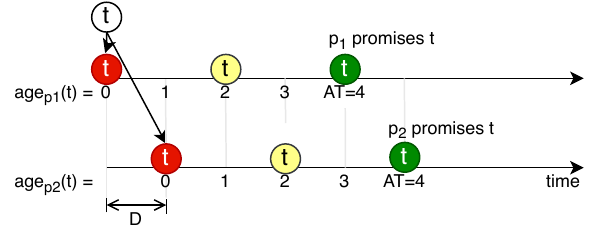}
\includegraphics[width=0.5\textwidth]{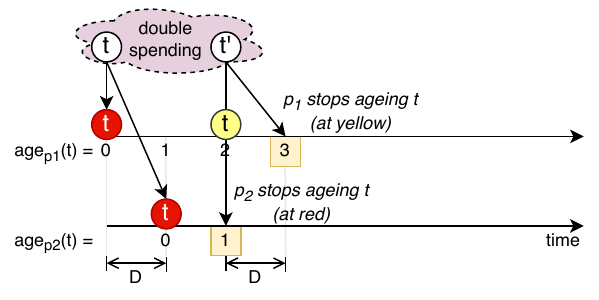}
\caption{Examples illustrating that the age of some transaction $t$ at different processes can diverge up to 2.}
  \label{fig:agebound}
  \squeezeup
\end{figure*}

Therefore, transactions issued by correct processes (which, by definition, do
not have double-spends) will always age up to green and hence be promised in the fast path at
every correct process. 
In contrast, a transaction $t$ that is issued as part of a double-spend attempt and is the first (among its
conflicting transactions) to be received by a process $p$ may age up
to red, yellow or green -- depending on how much time it takes until a double-spend $t'$ is subsequently 
received at $p$.
In any case, the first transaction received at $p$, $t$, will be added to the local mempool and thus $p$ 
will try to include $t$ in new blocks that $p$ 
tries to mine; in contrast, $p$ will reject $t'$ from the local mempool, hence
$p$ will never mine blocks with $t'$.


Since transactions are delivered at different times at distinct processes,
their final ages or colors are not guaranteed to be consistent across the system. 
Still, since the behaviour of the broadcast layer is bounded by a maximum delivery delay (D), 
the potential divergence is bounded and, as we prove next, always meet the 
\emph{bounded age consistency} property (see previous section).

This is illustrated in Figure~\ref{fig:agebound}. 
It may occur that $t$ ages up to green (thus, promised in the fast path) in some correct processes, 
but only up to yellow in others. 
Or that $t$ ages up to yellow at some correct processes, but only up to red in others.
Generalizing these examples, we can formulate and prove the following lemma.

\begin{lemma} 
\label{lemma:bounds}
If a correct process $p$ ages a transaction $t$ up to green (i.e., $age_p(t)=AT$) then: 
(a) for any correct process $q$,
$t$ ages at least up to yellow ($age_q(t)\geq AT-2$);
further, (b) for any transaction $t’$ that conflicts with $t$, then $age_q(t’)=\perp$ for every correct process $q$. 
\end{lemma}

\begin{proof} 
(1) Suppose, by contradiction, that either (a) $age_q(t)<AT-2$ or (b) $age_q(t)=\perp$; this implies that either (a) $q$ received $t’$ either less than $2D$ after $q$ first received $t$, or (b) that $q$ received $t’$ before $t$, respectively. 
Since $p$ received $t$ first and $t’$ more than $4D$ after (since $age_p(t)=AT$), then hypothesis (a) and (b) only occur in executions where $t$ or $t’$ was delivered at $q$ or $p$ (resp.) more than $D$ after they were broadcast by their issuer, which contradicts the assumption of a bounded delivery delay, $D$. (2) Suppose, by contradiction, that $age_q(t’)\neq\perp$. This implies that $q$ received $t’$ before $t$, which corresponds to hypothesis (b) above, which is impossible under the bounded delivery delay assumption.
\end{proof}

From sub-lemma (a) above, we directly obtain that the proposed ageing algorithm satisfies the \emph{bounded age consistency} property.
Due to space constraints, we ommit the correctness arguments for the remainder of this section.

\begin{algorithm}[t]
  \caption{Biased BBP commit (\name’s slow path)}\label{alg:slowpath}
  
    \SetKwFunction{FRRS}{RRS}

  	\SetKwProg{NewTx}{Upon receiving}{:}{}
    	\NewTx{chain $b_0, ..., b_n$} {
		\For{each $b_i (0\geq~i\geq~n)$ not yet in the local chain}
		{
		  \For{each transaction $t_i \in b_i$ that conflicts with
                  a previously received transaction $t$}
			{
				\If{$\FRRS{t} > n-i$}
				{
					reject chain suffix $b_i,...,b_n$\;\nllabel{alg:slowpath:rsscondition}
				}
			}
		}
	}

  \tcc{Simplest implementation}
  \SetKwProg{Fn}{Function}{:}{}
  \Fn{\FRRS{transaction $t$)}\label{alg:slowpath:rsssimple}}
  {
  	\eIf{$color(t) \in \{yellow, red\}$}
	{\Return{$C$}\;}
	{\Return{0}\;}
  }

  \tcc{Progressive variant (see \S\ref{sec:slicing})}
  \Fn{\FRRS{transaction $t$)}\label{alg:slowpath:rsssprogressive}}
  {
  	\tcc{Assuming $AT=2(C+1)$}
	{\Return{$\floor{age_p(t)/2}$}\;}
  }

\end{algorithm}

\begin{figure*}[t]
\center
\includegraphics[width=\textwidth]{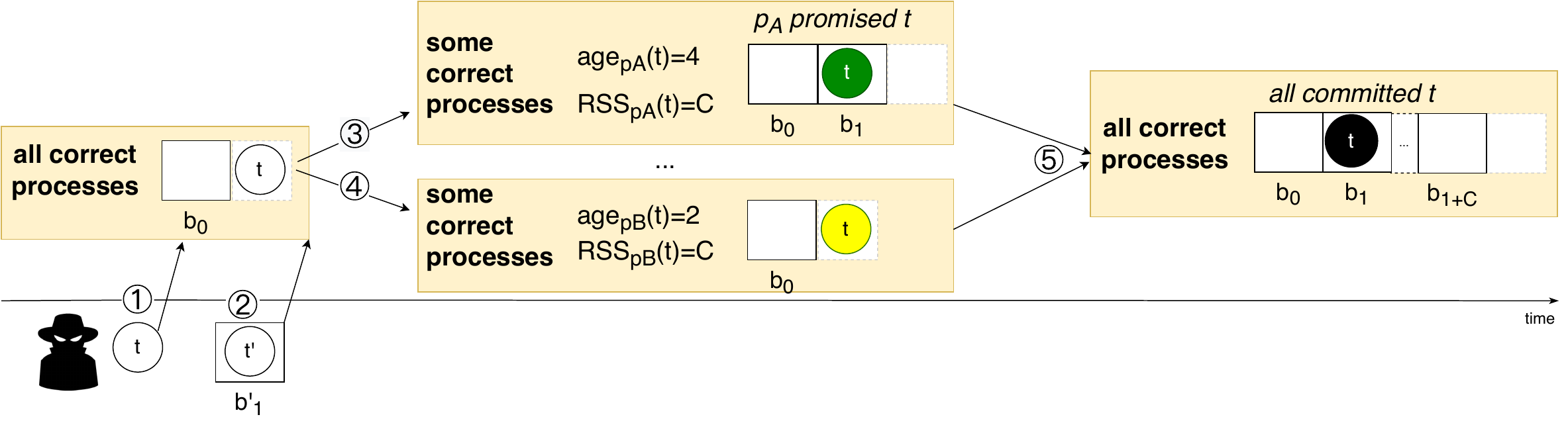}
\includegraphics[width=\textwidth]{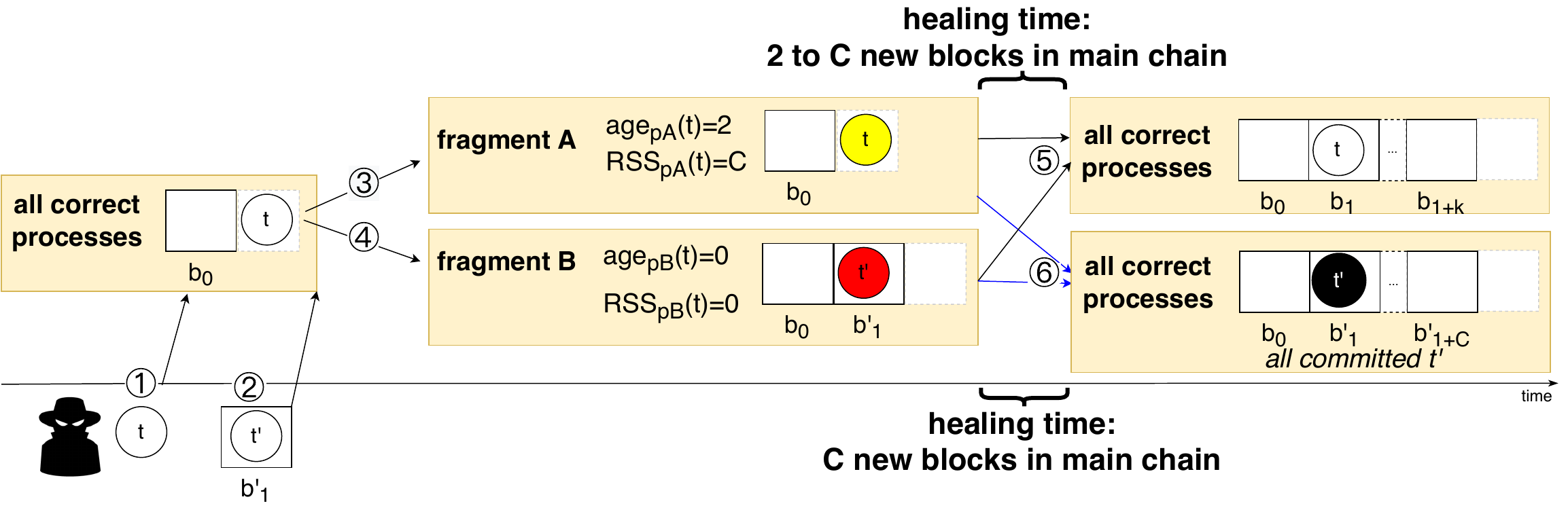}
\caption{Example of \name handling a transaction that is part of a double-spending attempt, with $AT=4$. Top: transaction ages up to green (thus, is promised) in at least one correct process. Bottom: transaction does not age up to green in any correct process.}
\label{fig:Parmenides-low-latency-scenarios}
\squeezeup
\end{figure*}

  \subsubsection{Slow path: Biased BBP commit}
We now focus on the slow path of \name, the biased BBP commit feature.
As outlined in \S\ref{sec:algoverview}, the biased BBP commit is activated by 
each process for a given transaction, $t$,
as soon as a $t$ ages up to yellow (at least).
According to the \emph{biased transaction selection} property (see \S\ref{sec:algoverview}), 
we expect that, if every correct process biases BBP towards $t$, then
BBP eventually commits $t$.

To implement the biased BBP commit feature, we add two restrictions on BBP's operation.
A first restriction is already described in Algorithm \ref{alg:ageing}.
From a set of double-spending transactions, only the first one
that is received at a process $p$ will be included in $p$'s mempool and
aged at $p$; any subquently received double-spends are rejected from $p$'s 
mempool.

We now introduce a second restriction, which affects BBP's CSR and is shown in
Algorithm \ref{alg:slowpath}.
For each new block $b_i$ in a new chain that a process $p$ receives,
$p$ checks whether block $b_i$ has any transaction that is a double-spending
of another transaction $t$ that $p$ had already received.
If so, the age of $t$ (at $p$) determines whether the offending block $b_i$
can be accepted or must be rejected.
If $t$ has aged at least to yellow at $p$, then $b_i$ is rejected, together with any subsequent blocks.

There is, however, an exception to this rule.
When $b_i$ is followed by at least $C$ blocks in the received chain,
\name no longer rejects $b_i$. We present the rationale behind this exception shortly.
Algorithm \ref{alg:slowpath} captures the above rule (and its exception) 
with
the \emph{required replacement suffix (RRS)} function. 
This function returns $C$ for those transactions
which the local process has aged to at least yellow; and 0 otherwise.


To illustrate, let us consider the first example in Figure~\ref{fig:Parmenides-low-latency-scenarios} (top). 
In this example, some correct processes (such as $p_A$) promise $t$,
as soon they age $t$ up to green
(\numcircledmod{3}).
From the \emph{bounded age consistency} property,
we know that even if, due to a double-spending transaction $t'$, other correct processes (such as $p_B$) 
might not age $t$ up to green, they are guaranteed to age
$t$ up to yellow (\numcircledmod{4}). 
Therefore, at that point, every correct process (either $p_A$ or $p_B$) biases BBP towards $t$.
According to the restrictions presented above, this means
that: 

\begin{enumerate}
\item $t$ has been included in every correct processes' mempool, thus they are all trying to 
mine a new block with $t$ (unless a block with $t$ is already
in their local chains), and any double-spend of $t$ is excluded from their mempools; and
\item all correct processes are rejecting any chains that contain a 
double-spend of $t$ \emph{unless the block containing such double-spend is suffixed by at least $C$ blocks}.
\end{enumerate}


Together, both observations imply that
the whole mining power of correct processes will be used towards appending
$t$ to their chains and continuing building on a chain with $t$.
The probability that an attacker is able to generate a chain containing a double-spending of $t$ 
followed by at least $C$ blocks is arbitrarily low. 
Therefore, the correct processes biasing BBP towards $t$ will never replace
the chain with $t$, which they are collectively building, by a fork with a double-spend of $t$. 
Hence, every correct process will eventually commit $t$ by BBP.


We also need to take other scenarios into account. 
These are illustrated in 
Figure~\ref{fig:Parmenides-low-latency-scenarios} bottom. 
Here, different correct processes only age $t$ up to either red or yellow. 
Those correct processes that age $t$ up to yellow, such as $p_A$ (\numcircledmod{3}),
bias BBP towards $t$, thus reject chains with any double-spending transaction $t'$ since $RRS(t)=C$.
In contrast, other correct processes that only age $t$ up to red, such as $p_B$ (\numcircledmod{4}),
remain open to accept chains with either $t$ or $t'$, since $RRS(t)=0$.

This situation can be problematic if the Byzantine process that issued the double-spend further manages
to extend the current main chain with a new block comprising $t'$.
In this situation, processes like $p_A$ will reject such new chain, while processes like $p_B$ 
will accept it (by the longest chain rule of BBP). 

We call this a \emph{fragmentation attack}.
It introduces an artificial fork that divides the mining power of correct
processes across two \emph{fragments}, denoted $A$ and $B$ in Figure \ref{fig:Parmenides-low-latency-scenarios}.
While the fragmentation persists, 
the attacker’s mining power will be temporarily closer to the 
largest correct fragment’s mining power, which may harm the robustness of BBP.



To mitigate the impact of a fragmentation attack, \name is able
to self-heal upon any successful fragmentation attack.
To understand how,
let us again focus on the bottom of Figure \ref{fig:Parmenides-low-latency-scenarios}.
A first scenario (\numcircledmod{5}) is when the fragment of correct processes biasing BBP towards $t$ are able to grow their chain faster than the opposite fragment (holding a chain with $t'$), but not vice-versa.
This fragment will be naturally healed by BBP's longest chain selection rule as soon as the main chain grows to be larger than the opposite fragment's chain. 

The second scenario (\numcircledmod{6}), however, may take longer to heal.
In this case, the fragment holding a chain with $t'$ happens to have the majority of correct processes. 
Therefore, it is likely that this chain will tend to grow faster than the chain with $t$. 
This is where the $C$ bound on $RRS$ becomes useful. The fragment holding a chain with $t$ will
accept the chain with $t'$ as soon as it becomes suffixed by $C$ blocks, which heals the fragmentation attack.
It should be noted that, in the scenario we are considering, no correct process has
promised $t$. Therefore, the system of correct processes can safely converge towards committing $t'$ and, thus, discard $t$.

\subsection{Trading fast path latency for robustness}
\label{sec:slicing}


The previous section showed that, with the minimal $AT=4$ value, a single fragmentation attack may require waiting up to $C$ block generation rounds 
to heal. 
In this section, we explain how \name can be configured to substantially mitigate the chances of success of the above attack.
The key insight is that, by increasing $AT$, \name can reduce the time it takes to heal fragmentation attacks from $C$ to the time it takes
for the fastest fragment's chain to generate 2 blocks, with high probability. 
This improved robustness comes at the cost of a higher fast path latency, since a higher $AT$ implies that transactions take longer to age up to green.


Before describing how \name can reduce the above vulnerability window,
let us recall the two fragmentation scenarios in Figure \ref{fig:Parmenides-low-latency-scenarios}, which assume $AT=4$.
In both scenarios, every correct process has received transaction $t$ and
an attacker is mining an alternative chain that contains a double-spend transaction $t'$.
While some correct processes will accept the alternative chain as long as it
meets the standard BBP requirements, others will be more reluctant and impose the additional constraint that, in the alternative chain, the block with $t'$  must be followed by at least $C$ blocks.
As shown in Figure~\ref{fig:Parmenides-low-latency-scenarios}, with $AT=4$,  $RRS_{p_i}(t)$ may differ between 0 and $C$ for different correct processes.
Intuitively, this upper bound on the divergence across $RRS_{p_i}(t)$ directly determines the time -- in terms of block generation rounds -- that a fragmentation attack may take to heal. 

If we can lower this upper bound, ideally down to 1, we can minimize the healing time.
We achieve this 
by extending the algorithm described in the previous section with a \emph{progressive} variant (ln. \ref{alg:slowpath:rsssprogressive}, Alg. \ref{alg:slowpath}) where each process gradually increments $RRS$
in $2D$ steps.
More precisely, we define $RRS_p(t)=\floor{age_p(t)/2}$.
To enable this progressive variant, while ensuring that $p$ only promises transaction $t$ once every correct process $p_i$ is biasing BBP towards $t$ 
(i.e., $RRS_{p_i}(t)=C$), we need to redefine $AT=2(C+1)$. 

As an example, suppose that process $p$ receives a transaction $t$.
Initially, $RRS_p(t)=0$. 
If $p$ does not observe any double-spend of $t$, $RRS_p(t)$ will grow to $1, 2, ...$ every $2D$, until it reaches $RRS_p(t)=C$ once $age_p(t)=C\times~2D$.
Finally, $2D$ later, $t$ ages up to green at $p$ and, hence, $p$ promises $t$ (at $age_p(t)=2(C+1)$).



\begin{figure*}[t]
\center
\includegraphics[width=\textwidth]{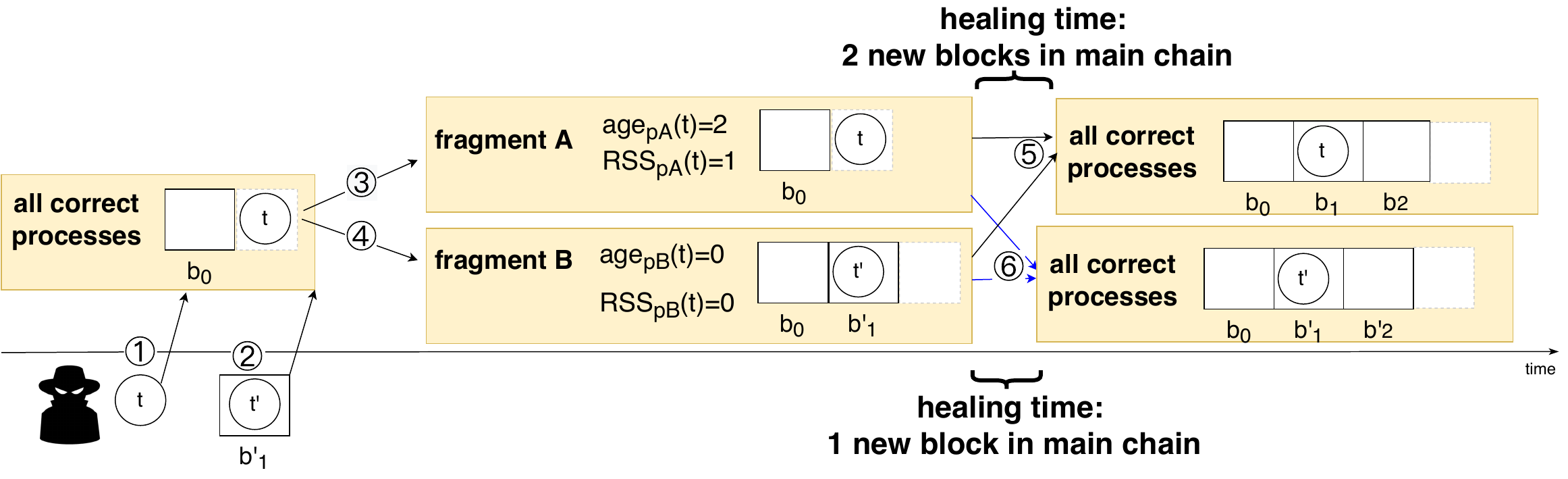}
\caption{Example of \name healing from a fragmentation attack, with its most robust configuration ($AT=2(C+1)$).}
\label{fig:Parmenides-robust-scenarios}
\squeezeup
\end{figure*}

To understand the impact of this new configuration, Figure~\ref{fig:Parmenides-robust-scenarios}
revisits the fragmentation scenarios from Figure~\ref{fig:Parmenides-low-latency-scenarios}.
As before, a successful
fragmentation attack has split the set of correct processes into two divergent fragments: fragment $A$, 
whose processes hold $t$ in their local ledger and now have $RRS_{p_A}(t)=1$ (instead of $RRS_{p_A}(t)=C$ as per the original formulation); and fragment $B$, whose processes have
$RRS_{p_B}(t)=0$ and, consequently, have accepted a chain with transaction $t’$, a double-spend of $t$, which the attacker produced and disseminated. 
On the one hand, a process in fragment $A$ does not accept the chain held by fragment $B$, 
since it holds a block with a double-spend of $t$ that is not suffixed by at least one block (due to
$RRS_{p_B}(t)=1$).
On the other hand, a process in fragment $B$ will not accept the chain of fragment $A$ 
since it has a lower height than fragment $B$'s main chain (due to the longest-chain rule).

Let us first consider the 
scenario where fragment $A$ is faster in extending its chain (\numcircledmod{5}) than fragment $B$.
In this case, as soon as $A$’s chain grows 2 blocks, processes in $B$  will start accepting that (according to the longest-chain rule).
Let us, instead, analyse the reverse situation, where fragment $B$'s chain grows faster (\numcircledmod{6}).  
The fragmentation may now heal as soon as fragment $B$ appends \emph{a single}
new block to its chain. 
These represent substantial improvements in healing time relatively to the $AT=4$ configuration, which could take up to $C$ blocks to heal in worst-case scenarios. 


We remark that these upper bounds are also observed in concurrent scenarios, where more than one block is produced either at the same fragment or across different fragments. 
Furthermore, we highlight that the 1 or 2 blocks needed to heal a
fragmentation can be produced by any process in either fragment, which means that the healing speed is determined by the aggregate resources of the correct nodes (i.e., not restricted to the largest correct fragment's resources).
These upper bounds on the healing time can be generalized to scenarios where: i) the fragmentation attack occurs when $t$ had a higher age; and, ii) the fragmentation attack simultaneously aims at multiple double-spending transactions. 

Summing up, \name can be configured between two extremes and offer a trade-off between fast path latency and robustness against fragmentation attacks.
Table \ref{table:tradeoff} summarizes the trade-off between both extremes of $AT$.


\begin{table}[t]
\footnotesize
\begin{tabular}{|p{1cm}|p{1.9cm}|p{2cm}|p{2.2cm}|}
\hline
\multirow{2}{*}{AT} & \multicolumn{2}{c|}{Upon ageing to green}                                                                                            & Upon fragmentation \\ \cline{2-4} 
                    & \begin{tabular}[c]{@{}c@{}}Promise\\ latency\end{tabular} & \begin{tabular}[c]{@{}c@{}}Speedup \\ (vs. BBP commit)\end{tabular} & Healing time       \\ \hline
$4$                 & $4D$                                                             & $C/(4r)$                                                              & up to $C\times~B$  \\ \hline
$2(C+1)$            & $2D(C+1)$                                                 & approx. $1/(2r)$                                                                & up to $2B$         \\ \hline
\end{tabular}
\caption{Trade-off between promise performance and fragmentation healing time in function of $AT$. The speedup is relative to BBP's commit latency of $C\times B$.
The $D/B$ ratio is denoted by $r$, which, by definition, is expected to be small \cite{garay2015bitcoin}.}
\label{table:tradeoff}
\squeezeup
\end{table}


\subsection{Ensuring causal order}
\label{sec:causal}

We now describe how \name handles causal dependencies consistently across the slow and fast paths, 
with the goal of ensuring promise and commit causal order (\S\ref{sec:promise}).
We hereafter assume that causal dependencies are specified in an additional transaction 
field using some suitable causality tracking mechanism.
%

Let us start by addressing promise causal order.
This is easy to ensure by adding a simple condition that checks causal dependencies before a transaction is
 promised upon ageing (ln. \ref{alg:ageing:promise}, Alg. \ref{alg:ageing}).
Concretely, when a transaction $t$ that has aged up to green at a given process, that process checks every 
causal dependency of $t$ before promising $t$.
If every such causal dependency is already promised or committed at that process, then it promises $t$.
Otherwise, $t$ is temporarily held in a local set of green transactions that still have 
causal dependencies awaiting to be promised/committed.

We also need to enforce commit causal order, which we achieve by extending \name with 
two additional restrictions to how each process manages its local chain.

The first restriction requires that the mempool comprises two distinct queues: a \emph{ready} and a \emph{pending} queue.
When a given process $p$ receives a transaction $t$ and is about to be included in the 
local mempool, an additional validation procedure will check every causal dependency of $t$.
If at least one causal dependency of $t$ is still absent from the \emph{ready} queue at $p$, 
$t$ is inserted in the \emph{pending} queue.
Otherwise, $t$ is added to the \emph{ready} queue, and any transactions in the \emph{pending} queue 
that causally depend on $t$ have their causal dependencies re-evaluated and, accordingly, 
moved to the \emph{ready} queue if $t$ was their last missing dependency.
When generating blocks, $p$ will only select transactions from the \emph{ready} queue. Therefore, any transaction in a block generated by a correct process 
is preceded (in the corresponding chain order) by its causal dependencies. 

As a second restriction, causal dependencies are checked for transactions in chains received from processes. 
If the causal dependencies are not satisfied in a received chain, that chain is discarded.

It is easy to show that the two above restrictions ensure that any transaction $t$ in 
a process' local chain is ordered after every causal dependency of $t$. 
Consequently, the committed prefix is causally ordered, which implies that transactions commit in causal order.

We remark that, while the algorithm described in the previous sections ensured that every transaction issued by a correct process would be promised at every correct process, this no longer holds when transactions have causal dependencies. 
As an example, suppose that a correct process $p$ has promised a transaction $t_1$ that was issued by a Byzantine process as part of a double-spend attempt; and, after observing $t_1$, $p$ issues $t_2$, which causally depends on $t_1$. 
Recall that, since $t_1$ is part of a double-spending attempt, $t_1$ may not be promised at some (or all) process.
In that case, some (or all) processes will only promise $t_1$. 
Hence, even though $t_2$ was issued from a process that had promised $t_1$ via its fast path, the processes that had not done so will delay $t_2$'s commit
until the outcome of $t_1$ is determined by the slow path. 

Finally, we note that, even though Nakamoto consensus defines a total order of transactions, it does not prescribe any specific transaction order to miners. In particular, miners are free to select which transactions to include in a block and the order in which they appear in the block. Our approach skews the order in which transactions are included in a block such that this order satisfies the causal order defined by the promises. But because Nakamoto consensus is itself oblivious to the concrete order of transactions within a block, our approach does not affect the correctness of the original protocol.



\section{Low-latency cryptocurrencies with \name}
\label{sec:cryptopromises}

As discussed in \S\ref{sec:consensuslesstx}, a notable application that can benefit from the weaker guarantees of the promise event 
provided by \name is a cryptocurrency. We now detail how.

A cryptocurrency can be abstracted as an instance of the asset-transfer object type defined by Guerraoui et al. \cite{guerraoui2019consensus}.
An asset transfer object maintains a set of accounts, where each account is associated with an owner client
that is the only one allowed to issue transfers withdrawing from this account. 
To do so, the owner client of an account $a$ can invoke a \emph{transfer(a,b,x)} to transfer $x$ from account $a$ to account $b$.
There is a second operation, \emph{read(a)}, which every process can invoke to read the balance of account $a$.

Traditional permissionless blockchains, such as BBP, implement the asset-transfer object type by relying on the consensus-based commit event, as follows. 

\begin{itemize} 

\item \textbf{\emph{transfer(a,b,x)}}.
When the process that owns account $a$ wishes to execute \emph{transfer(a,b,x)}, 
it reads the current balance of $a$ (see next) and, if the balance is sufficient, issues a
new transaction whose payload transfers $x$ from account $a$ to account $b$. 

\item \textbf{\emph{read(a)}}. The \emph{read(a)} operation is implemented by returning $a$'s balance from the state that results from the ordered execution of every committed transaction in the local chain.

\end{itemize}

Guerraoui et al. \cite{guerraoui2019consensus} prove that, in fact, the asset-transfer object type can be correctly implemented in
a consensusless fashion.
They also present (and prove correct) an actual consensusless implementation of the asset-transfer object type based on message passing.
The algorithm they propose relies on a secure broadcast layer that exposes a \emph{broadcast} and a \emph{deliver} event (for messages), 
while offering uniform reliable delivery with \emph{source order} guarantees despite Byzantine faults.
The complete algorithm by Guerraoui et al. defines which state each process maintains in order to know, which outgoing transactions have been issued by that process, as well as which incoming transactions have been delivered and validated at that process. Further, it defines how, based on that state, the causal dependencies field of a newly-issued transaction can be efficiently encoded.

We can port their approach to \name by replacing the underlying secure broadcast layer with \name.
Concretely, by simply replacing the \emph{broadcast} and \emph{deliver} events in Guerraoui et al.’s algorithm 
with the \emph{issue} and \emph{promise} events of \name’s interface.
The key insight to this transformation is that the properties that Guerraoui et al. require from the secure broadcast layer (namely, integrity, agreement, validity and source order \cite{guerraoui2019consensus}) are also satisfied by \name’s promises with high probability. 
We prove this later on this section.

Below, we present a high-level summary of the algorithm that results from porting Guerraoui et al.’s to rely on  \name’s interface. 
For lower-level details, we refer the reader to \cite{guerraoui2019consensus}. 

\begin{itemize} 

\item \textbf{\emph{transfer(a,b,x)}}. When a processes $p$ that owns account $a$ executes the \emph{transfer(a,b,x)} operation, it confirms that account $a$ has enough funds and, if so, issues a transaction $t$, whose payload describes the requested operation.
Further, the causal dependencies field of the new transaction $t$ is the set of transactions comprising:
i) every previously issued outgoing transfer transaction (i.e., transferring funds from $a$); 
and ii) every incoming transaction (i.e., transferring funds to $a$) already promised by process $p$. 

\item \textbf{\emph{read(a)}}. The \emph{read(a)} operation is implemented by returning $a$'s balance from the state that results from the ordered execution of every \emph{promised} transaction in the local chain.

\end{itemize}

Recall that, for most transactions, \name’s fast-path ensures that most correct processes are able to promise such transactions much sooner 
than the time that BBP slow-path takes to commit them.
Therefore, 
the above implementation of the asset-transfer object achieves important latency improvements. 

The above promise-based implementation is correct according to the specification of the asset-transfer object type \cite{guerraoui2019consensus}.
In other words, the above implementation correctly supports a cryptocurrency. 
The following lemma states this.

\begin{lemma}
The proposed promise-based implementation of \emph{transfer(a,b,x)} and \emph{read(a)} is a correct implementation
of an asset-transfer object type \cite{guerraoui2019consensus}.
\end{lemma}

\begin{proof}
To show that the above implementation is equivalent to Guerraoui et al.'s message passing asset-transfer object implementation, which was originally proposed and proved correct in \cite{guerraoui2019consensus}, we prove that \name's issue and promise events satisfy, with high probability, 
the properties of the secure broadcast layer underlying Guerraoui et al.’s implementation.
Next, we take the properties of 
the broadcast layer that underlies Guerraoui et al.'s algorithm (namely, integrity, agreement, validity and source order \cite{guerraoui2019consensus})
and reformulate them by renaming the \emph{broadcast} and \emph{deliver} events
by the \emph{issue} and \emph{promise} events. 
Then, we prove that the resulting properties are satisfied by \name.


\noindent\emph{Integrity: a correct process promises a transaction $t$ from a process $p$ at most once and, if $p$ is correct, only if $p$ previously issued $t$. }
This is ensured since transactions in BBP are digitally signed and carry a unique identifier. 

\noindent\emph{Agreement: if processes $p$ and $q$ are correct and $p$ promises $t$, then $q$ promises $t$.}
Let us recall that, if a correct process $p$ promises some transaction $t$, then $p$ eventually commits $t$ (by the \emph{eventually committed upon promised} property). Hence, any other correct process $q$ also eventually commits $t$ (by BBP's persistence property), 
thus, by definition of promise, $q$ also promises $t$.

\noindent\emph{Validity: if a correct process $p$ issues $t$, then $p$ promises $t$.}
This is ensured since any transaction $t$ issued by a correct process $p$ is eventually committed by $p$ (by BBP's liveness property) therefore, by definition of promise, $p$ also promises $t$.

\noindent\emph{Source order: 
if $p$ and $q$ are correct processes and both promise transactions $t$ and $t'$, both issued by the same
process $r$, then they do so in the same order.}
Let us first consider the case where $t$ and $t'$ have distinct sequence numbers. Therefore, the promise causal order property guarantees that both $p$ and $q$ will promise both transactions
in order defined by their sequential numbers.
Let us, instead, assume by contradiction that $t$ and $t'$ have the same sequence number and are promised by $p$ and $q$. 
Therefore, by the \emph{eventually committed upon promised} property, $p$ and $q$ will eventually commit both transactions.
Still, by definition, distinct transactions with the same sequential number are conflicting and, therefore, at most one can commit. 
This contradicts the initial assumption.
\end{proof}

We conclude with two final remarks. First, although Guerraoui et al.’s algorithm was originally proposed in the context of permissioned systems, 
adapting it to exploit the primitives of \name enables it the work in permissionless environments.
Furthermore, while the original proposal supported a stand-alone cryptocurrency system, the above adaptation to \name 
integrates the low-latency cryptocurrenty in a richer ecosystem where other applications with stronger consistency requirements
may also co-exist. 
For instance, this hybrid consistency ecosystem enables processes to issue smart contract transactions (via the issue/commit interface), whose execution costs are charged from cryptocurrency accounts which may receive incoming transfers as defined above (via the issue/promise interface).

\section{Evaluation}
\label{sec:evaluation}
\label{sec:admitandcommittime}

In this section, we evaluate \name with the goal of answering the following
questions: i) what are the latency improvements that \name brings to applications with different consistency needs, and ii) what is the impact of fragmentation attacks?

As we further detail below, our experimental evaluation is performed in an environment as close as possible to a real deployment. Namely, it relies on the real code of the reference Geth implementation of 
Ethereum \cite{Geth2017}, which we extend to implement \name, running in a system of 500 processes.
Adopting a similar methodology as related works (e.g., \cite{Eyal2016,li2018conflux,Yu2020ohie}), the workload is derived from a real trace of Ethereum,  PoW mining is simulated, and a realistic geo-distri\-buted network is emulated using a state-of-the-art emulator~\cite{Kollaps}.

Next, we detail the evaluation scenario, metrics and discuss our results.


\begin{table*}[]
  \center
\footnotesize
\begin{tabular}{|l|l|}
\hline
Number of processes   & 500\\ \hline
Mining power & 24.0\%, 21.3\%, 13.2\%, 12.1\%, 5.7\%, 1.9\%, 1.8\%, 1.5\%, 1.4\%, \\ 
& 1.3\%, 1.1\%, 1.0\%, 1.0\%, and 0.026\% for each remaining process\\ \hline\hline
Average block generation time ($B$) & 20 s\\ \hline
Average transaction generation time & 125 ms\\ \hline
Commit threshold ($$C$$) & 12\\ \hline\hline

Average message delivery delay & 120 ms\\ \hline
Maximum delivery delay ($$D$$) & 960 ms\\ \hline\hline

\end{tabular}
\caption{Summary of experimental parameters.}
\label{table:experiments}
\end{table*}

\begin{figure}[!t]
  \centering
  \includegraphics[width=0.48\textwidth]{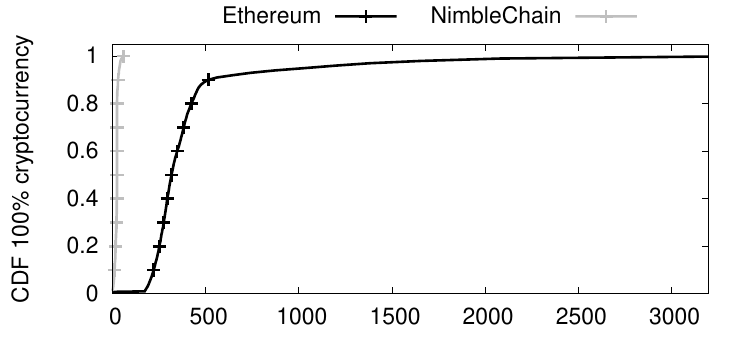}
  \includegraphics[width=0.48\textwidth]{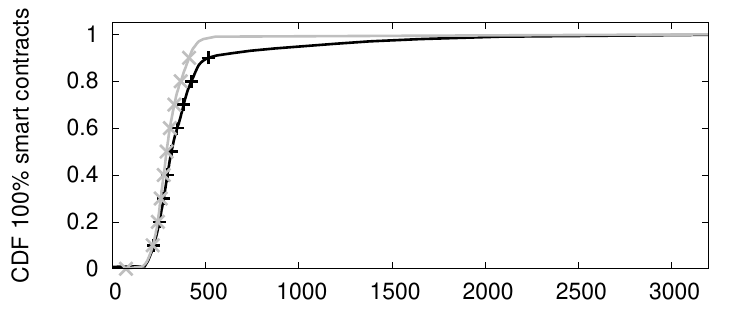}
  \includegraphics[width=0.48\textwidth]{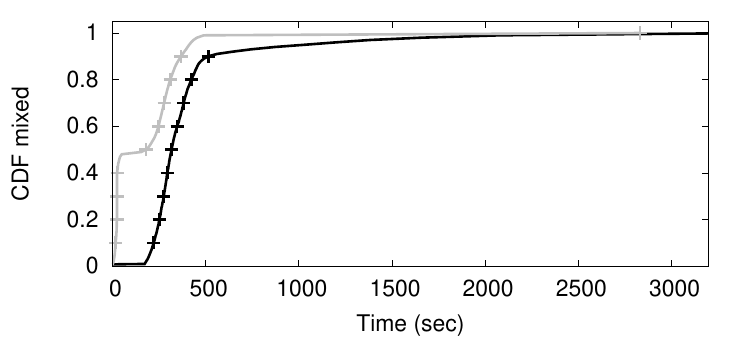}

  \caption{Transaction latency: 100\% cryptocurrency transactions (above);
           100\% smart contract transactions (middle);
          mixed cryptocurrency and smart contract transactions}
  \label{fig:cdfadmiteth}
\vspace{-.5cm}
\end{figure}



\subsection{Implementation, deployment and workload.} 
\label{sec:implementation}

We implemented \name as an extension of Ethereum using the reference Geth implementation~\cite{Geth2017}.
The implementation required 
$\approx 3000$ new lines of code to implement \name and changing $\approx 300$ lines of code in Geth.

For experimental purposes, we also developed a custom client that, using the regular API, injects transactions in the system according
to real transaction traces from Etherscan~\cite{Etherscan2019} 
(a sample of transactions from block 5306612 up to block 6222336).
The workload has no information about causality (aside from the implicit dependencies among transactions issued from the same account) and hence we 
set each transaction to depend on the most recently promised transaction -- and, transitively, from the causal dependencies of that transaction. 
Note that this is a conservative choice, since it increases the probability that some processes will have to wait for causal dependencies. 

Table \ref{table:experiments} summarizes the main experimental parameters. To run experiments with a large number of processes we replaced the PoW component 
with a probabilistic mining selection process that follows a Poisson distribution and mimics the block production distribution.
To reflect the non-uniform mining power distribution of today’s mainstream permissionless blockchains, 
we allocated the estimated mining power of the top-13 most powerful mining pools of Ethereum to a subset of 13 processes, according to \cite{Etherscan2019} (as detailed in Table \ref{table:experiments}).
The remaining mining power was uniformly distributed across the remaining 487 processes.

We adjusted the block production probability to mimic Ethereum's rate of 3 blocks per minute~\cite{Etherscan2019}.
Moreover, we extended both \name and Ethereum implementations to log transaction events such as 
generation, reception, dissemination and block events (such as 
insertion to the local chain and forks) to allow \emph{a posteriori} offline 
processing for our evaluation.
We use the same codebase, client and PoW component for the \name and Ethereum.
Every correct process is parameterised with $C=12$, the current standard commit threshold in Ethereum~\cite{Wood2014}.

We ran each experiment for one hour with 500 processes for both \name and Ethereum, using 5 machines equipped with a mix of Intel(R) Xeon(R) CPUs. 
We empirically found this configuration of machines to be able to accommodate 500 processes without becoming overloaded and hence compromising the fidelity of the results.
The processes run on an emulated network using Kollaps~\cite{Kollaps} with internet latencies that model the geo-distributed nature of permissionless blockchains.
As suggested by previous measurement studies~\cite{Sirer2018, silva2020impact}, we used an average latency value of $120ms$, and a conservative value for $D = 8 \times 120ms = 960ms$.
Each node started with the same local chain, consisting of a single genesis block.
We injected 8 transactions per second, 
as common in
Ethereum \cite{Etherscan2019}. 
All results are the average of 5 runs.

\subsection{Promise and commit latencies} 

In this section, we study \emph{transaction latency} as perceived from a process $p$, which we define as the time from the moment a given transaction $t$ is issued (at some process, not necessarily $p$) and the moment $p$ triggers the event that is required by the application semantics associated with $t$ (i.e., either promise or commit).
We evaluate transaction latency for two transaction types,  with distinct consistency needs: cryptocurrency transactions, which only require promises; and smart contract transactions, which need to be committed. 

We consider 3 scenarios: i) 100\% cryptocurrency transactions;
ii) 100\% smart contract transactions; iii) and a mixed ratio of 44\% cryptocurrency and 66\% smart contract transactions as observed in Ethereum~\cite{silva2020impact}. 
We evaluate these 3 scenarios considering 
no Byzantine behaviour.  (We study Byzantine behaviour in the next section.)

We consider \name configured with $AT=2(C+1)$, the most robust configuration.
For brevity, we do not evaluate \name configured with
$AT=4$, which is a less robust configuration that would present
even lower values for promise latency. 

Figure~\ref{fig:cdfadmiteth} 
presents our results.
As expected, the average transaction latency for cryptocurrency transactions
is around one order of magnitude lower with \name (promise latency) than with Ethereum (commit latency),
as shown in Figure~\ref{fig:cdfadmiteth} (top).

Smart contract transaction commit times with Ethereum are very similar to \name's, which suggests
that the overhead of \name on the underlying protocol is negligible, as shown in Figure~\ref{fig:cdfadmiteth} (middle).
\name seems to slightly outperform Ethereum at the head, due to slightly different
transaction ordering criteria in the mempool. As an example, causal dependencies are ordered before
a dependent transaction regardless of their price. Hence, such dependencies may reach the blockchain
earlier in \name than in Ethereum.

Figure~\ref{fig:cdfadmiteth} (bottom) depicts the mixed scenario and shows an interesting trend.
Cryptocurrency transactions continue to perform much faster in \name (promise latency) than in Ethereum (commit latency), while smart contracts commit at roughly the same pace with the results showing a clear inflection point between each transaction type.

\begin{figure}[t]
  \includegraphics[width=0.5\textwidth]{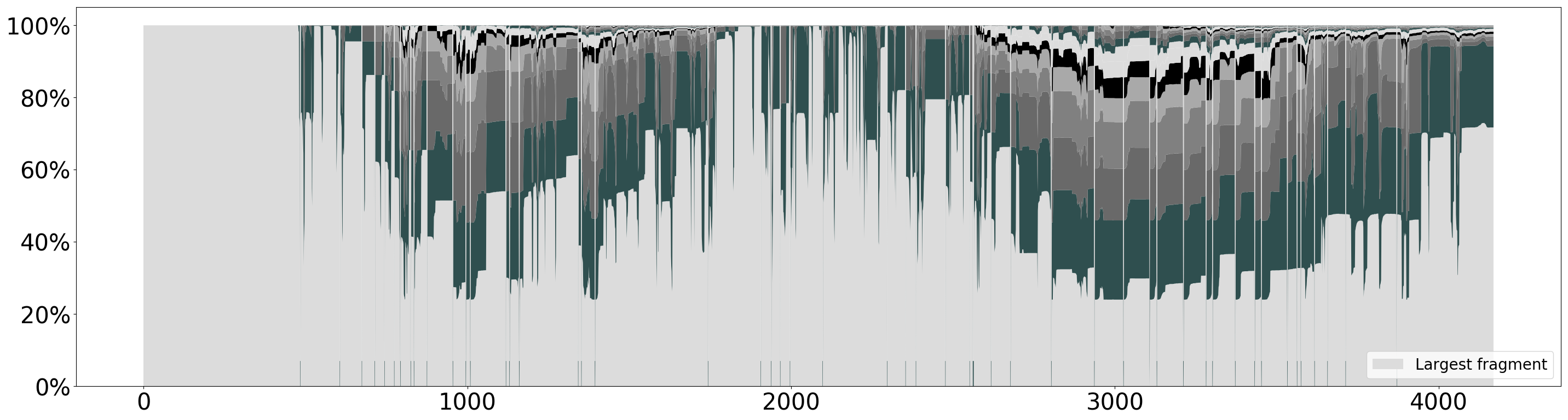}
    \includegraphics[width=0.5\textwidth]{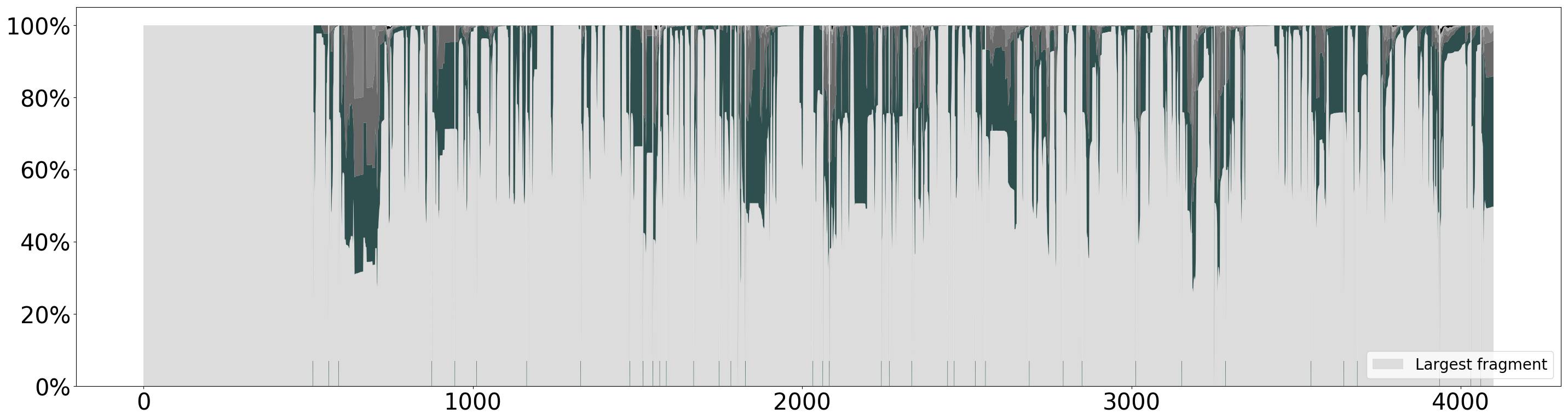}
      \includegraphics[width=0.5\textwidth]{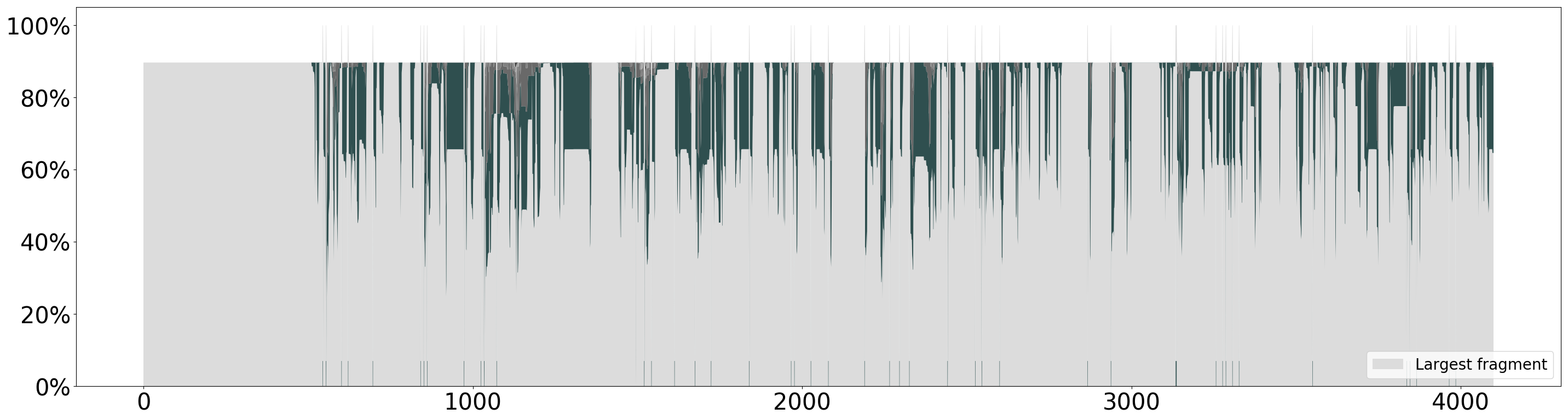}
  \caption{Fragment adoption during a fragmentation attack, comparing \name with $AT=4$ (top) and $AT=2(C+1)$ (middle) against the Ethereum baseline (bottom, no fragmentation attack) The y-axis
expresses the percentage of resources owned by each fragment of correct processes, where 100\% in the same colour (light grey) means that every correct process’  local chain is identical. The x-axis represents time, in seconds.\
}
  \label{fig:fragmentation}
\end{figure}

\begin{figure}[t]
\includegraphics[width=0.5\textwidth]{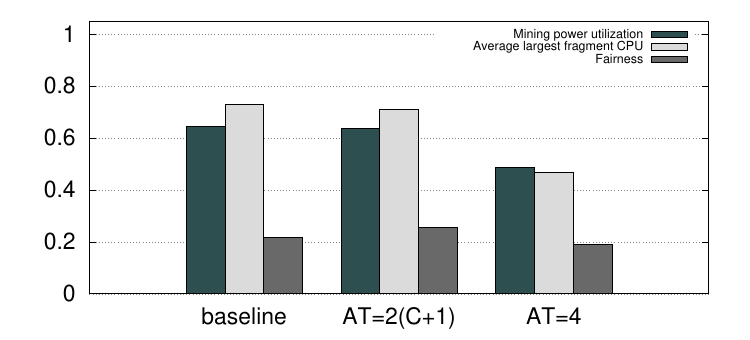}
\caption{Mining power utilization, average largest fragment CPU and fairness upon fragmentation attack.
}
\label{fig:mpu}
\end{figure}	

\subsection{Fragmentation attacks} 

To perform a fragmentation attack for some period $T$, an attacker must allocate a portion $x$ 
of its mining power to that purpose during $T$.
As a return from such investment, the attacker will reduce the average mining power of the largest correct fragment by $y$ during $T$.
To be profitable to the attacker, $y$ must be higher than $x$. 
Otherwise, a rational attacker will not have any tangible benefit, and thus will not carry on the attack.
The main goal of this section is to study the impact of fragmentation attacks, and consequently answer the above question.

Recall that a fragmentation attack involves broadcasting a transaction $t$ and immediately starting to produce a block with a double-spend $t'$.
Later, the attacker broadcasts $t'$ in an attempt that different correct processes will receive it at different ages of $t$.
This will result in some correct processes choosing $RRS_p(t)=a$ and others $RRS_q(t)=a+\delta$, where $a$ denotes the 
age that $t$ had when the former received $t’$ and $\delta$ depends on the $AT$ parameter (as discussed in \S\ref{sec:slicing}).

We empirically found that the most advantageous case for the attacker was with $a=0$, i.e., some correct processes $p$ receive $t'$ when $age_p(t)=0$, while other processes $q$ receive $t’$ when $age_q(t)>0$ (i..e, either 1 or 2). Hence, in this analysis, we consider $a=0$. 



As detailed in \S\ref{sec:slicing}, the divergent fragments of correct processes eventually converge when at most 2 new blocks are generated from either fragment (and delivered), when $AT=2(C+1)$.
In contrast, when $AT=4$, the divergent fragments of correct processes may require $C$ blocks to converge.
In the meantime, the attacker may
perform a burst of consecutive fragmentation attacks with the
intention to keep the system fragmented for a longer period,
by creating a new fragmentation before the current one is healed.
In these experiments, we assumed the attacker holds $24\%$ of the resources, which corresponds to the largest mining pool in Ethereum~\cite{silva2020impact}.
The attacker continuously 
performs a fragmentation attack to weaken the aggregate resources owned by correct processes,
by making them adhere to distinct fragments. 
The attacker uses its total mining power to perform the fragmentation attack series ($x=24\%$).
For brevity, we omit scenarios where the attacker only spends a fraction of its total mining power ($x<24\%$). 
Our analysis of lower values of $x$ yielded similar observations as the ones that we present next.

Further,
we artificially instrumented the delivery protocol to ensure a \textit{fragmentation ratio} that is favorable to the attacker: 80\% of nodes randomly
considered $RRS_p(t_1)=0$ and could accept a malicious block immediately,
while the remaining 20\% considered $RRS_p(t_1)=1$ for $AT=2(C+1)$ or $RRS_p(t_1)=C$ for $AT=4$. 
In a real scenario,
it would be very unlikely for the attacker to be able to reach such
precise fragmentation, since he does not control the network latency.
We tested other values for the fragmentation ratio (e.g. 50\%-50\%, 70\%-30\% and 60\%-40\%), which 
yielded similar conclusions and hence we omit those results.

Figure~\ref{fig:fragmentation} shows the results as stacked histogram. 
The y-axis expresses the percentage of resources owned by each fragment and the x-axis represents time, in seconds.
The light colored layer represents the percentage of resources owned by honest processes in the largest (i.e., with the highest mining power) fragment at a given moment. 
It is worth noting that, even in the baseline, the mere generation of a new block by a correct process (with or without a concurrent fork) causes short periods where the correct processes are not perfectly synchronized (i.e., in the same local chain state), which quickly end as the system synchronizes again.
Note that this is intrinsic to BBP.

As expected, the results for $AT=4$ (Figure~\ref{fig:fragmentation} top) show a long-lasting fragmentation since processes require a long suffix of $C$ blocks before accepting a chain with $t_1'$. 
For $AT=2(C+1)$ (Figure~\ref{fig:fragmentation} middle) we observe a much faster healing since processes require only $1$ block before accepting a chain with $t_1'$.
This confirms that $AT=4$ trades a decreased latency in the fast path for less robustness against
fragmentation attacks, while $AT=2(C+1)$ provides more robustness at the expense of a higher latency in the fast path. 

To conclude our analysis of the fragmentation attack, we evaluate its
impact on two metrics originally proposed by Eyal et al.~\cite{Eyal2016}.
The mining power utilization (MPU) is the ratio between the aggregate work of the main chain and all produced blocks.
Fairness is the ratio between the number of blocks generated by the largest honest miner and all produced blocks. In a fair system, the fairness ratio should be identical to the mining power of the reference miner, which holds 21.3\% mining power in this experiment (corresponding to the 2nd most resourceful miner). 
We also measure the average CPU power of the largest fragment.

The results are presented in Figure~\ref{fig:mpu}.
As it is possible to observe, the MPU decreases from its baseline value when an attack is carried out.
Further, MPU decreases as $AT$ decreases.
The average CPU owned by the largest fragment during the attacks also decreases as $AT$ decreases. 
These are expected results, since the mining power is scattered among fragments.  

Most importantly, the decrease in both metrics is considerably lower than the mining power the attacker invested to put the attack in practice. 
This is true even for the least robust variant of \name, $AT=4$. 
This means that the attacker does not reach a \emph{break even} point. Therefore, the analyzed attack is not profitable for a rational attacker.
Instead of using his mining power to slow down the correct system's ability to advance the main chain (through fragmentation), it would be
more profitable for the attacker to employ the same mining power to accelerate the generation of his malicious chain (an attack vector that is possible in standard BBP). 

Finally, and as expected, the fragmentation periods cause \textit{fairness} deviations in \name. 
While vanilla Ethererum is the closest to the desired fairness target (ideally, 21.3\%), the fairness of \name’s variants are still within a 10\% distance from the ideal target. 


\section{Discussion}
\label{sec:discussion}



\textbf{Assumptions on network propagation.} The design of \name depends on the assumption of a well-known maximum delivery delay, $D$.
Of course, \name may behave incorrectly if the $D$ assumption is not met by the underlying network. 
More precisely, a period of arbitrary propagation delays (beyond $D$) can violate our assumption that the age and $RSS$ that 
two correct processes assign to some transaction do not diverge by more than 2 and 1 (respectively).
An arbitrary divergence across the ages each process sees may lead to pathological situations where two double-spending transactions are able to 
successfully age at distinct processes. Further, an arbitrary divergence in $RSS$ may compromise the ability of \name to heal upon a fragmentation attack. 

We remark that a maximum network propagation time is a standard assumption in permissionless blockchains. For instance, every permissionless blockchain mentioned in \S\ref{sec:rw} relies on this assumption.
Further, different works have studied vulnerabilities that are possible if the $D$ assumption is violated. 
For instance, an unexpectedly high block propagation delay may slow down the time the system converges upon forks, which provides
an advantage to a resourceful attacker to temporarily benefit from lower mining power utilization from the correct processes~\cite{garay2015bitcoin}.
An eclipse attack \cite{Marcus2018} can isolate a subset of correct processes from the remaining correct system -- which can be translated to arbitrarily increasing $D$ from the outside to the eclipsed partition --, allowing a resourceful attacker to control the evolution of the blockchain within the partition.

Still, we acknowledge that \name might be more sensitive than BBP to smaller sporadic violations of $D$, and/or smaller periods where $D$ is violated, during which \name's properties are violated but BBP's hold.
This is a natural consequence of a protocol 
that reaches a decision over a shorter time window.
This consequence is shared with every proposal in \S\ref{sec:rw} that provides lower commit latencies than BBP.

\textbf{Performance with higher delivery delays.} In absolute terms, \name depends on $D$ being low enough to ensure a low fast path latency.
Hence, one may rush to the conclusion that the speedup that \name introduces relatively to the baseline BBP-based permissionless blockchain
depends on $D$.
However, that is not correct. Let us suppose that \name was used in a low-quality network whose $D$ was much higher than the one that previous studies
find in real permissionless blockchains~\cite{Kogias2016,Kogias2019}.
In that case, the PoW difficulty (and, hence, the block generation time, $B$) would need to be adjusted accordingly, in order to keep the rate $D/B$ low enough to ensure the (probabilistic) correctness of BBP.
This readjustment of $D$ and $B$ would not only increase the fast path latency of \name, but also the commit latency of BBP. Therefore, the speedup of \name would remain the same. 
This conclusion is in line with Table \ref{table:tradeoff}, which shows that the speedup of \name depends on $f$ and $C$, not on $D$.

\section{Related work}
\label{sec:rw}

In recent years, many proposals have arisen that improve the performance of permissionless blockchains based on Nakamoto consensus.
In the vast majority of such proposals, meaningful consistency guarantees are only provided once a block is followed by enough (i.e., $C$) blocks in the ledger, 
similarly to Nakamoto's original proposal. Therefore, such proposals do not escape the well-studied  lower bound of commit latency in Nakamoto consensus~\cite{garay2015bitcoin}.
Instead, they provide improvements to throughput and/or energy efficiency, but not on commit latency.
In contrast, \name~offers a hybrid consistency model in which weaker consistency guarantees are provided by a new promise event, enforced by a different protocol that coexists with Nakamoto consensus. Consequently, promise latency is not dictated by the lower bounds of Nakamoto consensus' commit latency.

In this section, we start by surveying related work that improves permissionless blockchains based on Nakamoto consensus. 
Later, we describe work on more disruptive approaches that, while falling outside the domain of Nakamoto-based permissionless blockchains, 
are related to our work.


\textbf{Improvements over longest chain rule.}
GHOST \cite{ghost}, partially implemented in Ethereum \cite{Wood2014}, improves Nakamoto's original longest chain rule by allowing all blocks generated by honest participants to contribute to the commit of the main chain. This enables convergence even with higher block generation rates.
A different approach, followed by inclusive blockchain protocols \cite{DBLP:conf/fc/LewenbergSZ15} and PHANTOM \cite{cryptoeprint:2018:104},
organizes blocks as
directed-acyclic graphs (DAG) of blocks instead of a totally-ordered list in order to optimize performance.
More recently, Conflux \cite{254398} combines the main principles behind GHOST and DAG-based solutions in an adaptive fashion to 
provide them with better liveness guarantees. 
Despite the significant throughput gains (e.g., Conflux is able to improve GHOST's throughput by 32x),
these approaches still rely on consensus as the only path to commit. 
Consequently, they only bring modest commit latency savings (e.g., 25\% latency gains in Conflux with respect to GHOST).
\name can be plugged to any of these systems and enhance them with substantially lower commit latencies, while retaining their throughputs.

\textbf{Hierarchical and Parallel Chains.}
Alternative blockchain organizations include hierarchical and parallel chains. 
In Bitcoin-NG \cite{Eyal2016}, key blocks are generated at a similar rate as Bitcoin.
Still, in-between two key blocks, the miner of the previous key block can generate many microblocks that contain transactions. 
FruitChain \cite{10.1145/3087801.3087809} adopts a similar hierarchical approach. 
OHIE uses parallel instances of BBP and then deterministically sorts blocks to reach a total order~\cite{Yu2020ohie}.
In all these proposals, the total order of the main chain is determined by only a 
fraction of blocks (key blocks).
Hence, the remaining blocks (microblocks), which carry the actual transactions, can be safely generated at much higher
rates than BBP allows, thus increasing throughput.
Unlike \name, these proposals focus on improving throughput, not commit latency.
For example, Bitcoin-NG does improve the time it takes the system to agree on Bitcoin-NG's microblocks but does not improve commit latency~\cite{Eyal2016}, while OHIE's average commit latency is around 10 minutes \cite{Yu2020ohie}.
One exception is Prism \cite{10.1145/3319535.3363213}, which supports low-latency and high-throughput honest transactions by resorting to parallel voting chains, which determine the total order of blocks in the main chain. Still, their simulation-based evaluation results are around 2x higher than the causal commit latency of \name's most robust configuration  (40 to 58 sec with $\beta=0.3$ \cite{10.1145/3319535.3363213}).
Since all the above proposals are optimizations over BBP's foundations, \name can supplement any of them with the low-latency of our promise fast path.

A notable alternative is to organize transactions as a directed acyclic graph (DAG), which has been proposed in IOTA's Tangle system \cite{SILVANO2020307} to achieve 
improved performance, including lower latency. However, the safety of this solution depends on an honest central point of control, which is at odds 
with the decentralized and permissionless nature of blockchains.


\textbf{Sharding.}
Systems such as Elastico \cite{10.1145/2976749.2978389}, OmniLedger \cite{kokoris2018omniledger}, Rapid-Chain \cite{10.1145/3243734.3243853}, Monoxide \cite{wang2019monoxide}, Ethereum 2.0~\cite{EthSharding2018} or Tao et al.'s proposal \cite{9101451} rely on multiple parallel blockchains cooperating via sharding, where 
a small committee maintains each shard.
This approach achieves substantial throughput gains (up to thousands of transactions per second), but at the cost of security, since
the smaller shards are vulnerable to powerful attackers.
Since sharded proposals typically assume (multiple) BBP-based blockchain instances, \name can be generalized to supplement sharded proposals with a low-latency promise fast path. 
Some sharded proposals have also been shown to achieve comparable commit latency savings as \name (e.g., \cite{kokoris2018omniledger,10.1145/3243734.3243853}). Still, such results are possible in networks with much lower RTT than the one considered in our paper and only in specific best-case workloads.

\textbf{Proof-of-X alternatives.}
Another research avenue has proposed permissionless consensus algorithms that replace PoW with energy-efficient alternatives,
such as Proof-of-Stake~\cite{Kiayias2017,zamfir2017casper,Gilad2017,Asayag2018}, Proof-of-Space~\cite{dziembowski2015proofs,BeyondHellman:2017} or Proof-of-Elapsed-Time~\cite{chen2017security}. 
Among such proposals, some are still based on a variant of BBP, despite replacing the PoW leader election component by a PoX alternative (e.g., \cite{Kiayias2017,zamfir2017casper,dziembowski2015proofs,BeyondHellman:2017}). Therefore, \name's fast path can be integrated onto such proposals.

\textbf{Blockchains based on Byzantine Fault Tolerance (BFT).} 
A relevant body of work leverages BFT protocols, executed among small committees of processes, to improve the performance of permissionless blockchains.
Proposals such as ByzCoin \cite{10.5555/3241094.3241117}, Thunderella \cite{DBLP:conf/wdag/PassS17} and Solida \cite{Abraham2017} combine BBP with BFT protocols.
Other proposals such as Algorand \cite{Gilad2017}, HoneyBadger \cite{Miller2016}, and
Stellar \cite{10.1145/3341301.3359636} are more disruptive. 
These approaches can achieve comparable commit latencies as \name's most robust configuration (e.g., in Algorand \cite{Gilad2017}'s best-case setting, 22 sec).
However, they significantly change the trust assumptions of permissionless blockchains, as BFT consensus
requires that 2/3 of the validators must be trusted.
Further, most of these proposals are highly disruptive and, as such, cannot incrementally extend existing mainstream permissionless blockchains. 


\textbf{Layer-2 proposals.}
Layer-2 proposals employ an additional protocol layer that handles (and commits) transactions and use the permissionless blockchain as a backend anchor to ensure consistency in the presence of malicious behaviour. In that sense, \name fits into this broad category.
Among the most relevant proposals, so-called off-chain solutions such as the Lightning Network~\cite{lightning} and FastPay~\cite{hao2018fastpay} 
rely on a separate network of payment channels and allow two or more parties to exchange currency without committing in the blockchain.
However, these proposals have important shortcomings. 
They work at the expense of temporarily locking payment guarantees (often called \emph{collaterals}) in the blockchain if a party misbehaves.
While proposals based on payment networks are not tailored to unidirectional payment flows (as typical in retail payments from customers to merchants \cite{mavroudis2020snappy}), alternatives based on payment hubs \cite{perun,tumblebit} either impose trusted entities or increased locked funds.

More recently, Snappy \cite{mavroudis2020snappy} proposes a novel on-chain smart-contract-based alternative that mitigates the above-mentioned shortcomings and achieves payment commit latency in the order of a few seconds. Still, Snappy has important scalability limitations in the number of 
payment recipient processes (up to 200 statekeeping merchants \cite{mavroudis2020snappy}). Further, since Snappy payments require smart contract invocations, they cost 8x more than simple transactions (in Snappy's Ethereum-based implementation \cite{mavroudis2020snappy}).
In contrast, \name neither requires collaterals, nor restricts scalability, nor increases transaction cost.


\textbf{Weakly-consistent blockchains.}
Like \name, some proposals attempt to obtain partial orders instead
of total orders for cryptocurrency transactions. 
In the permissionless world, notable proposals include 
SPECTRE \cite{sompolinsky2016spectre}, 
TrustChain \cite{otte2017trustchain}, 
ABC \cite{sliwinski2019abc}, Avalanche \cite{rocket2020scalable} and Pastro \cite{DBLP:conf/wdag/KuznetsovPPT21}. 
In the context of permissioned blockchains, 
Astro \cite{collins2020online} exploits Byzantine reliable broadcast \cite{malkhi1997high} to build a payment system.
All these proposals exploit the fact that cryptocurrecy transfer transactions do not need to be totally-ordered, hence
can be managed by weaker primitives than consensus.
%
%
Like \name's promise fast path, the above proposals can serve the weak consistency needs of some applications such as cryptocurrencies. 
Still, these proposals cannot directly support general-purpose blockchains, whose application ecosystem 
comprises applications with weaker consistency demands (such as cryptocurrencies) and strong sequential consistency (such as
most smart contracts). In contrast, \name's hybrid consistency model is tailored to such mixed ecosystems. 

\textbf{Permissioned blockchains.} Permissioned blockchains (e.g., \cite{10.1145/3190508.3190538, 10.1145/3299869.3319889, 10.1145/3477132.3483584}) have emerged as a considerably more efficient alternative to the permissionless counterpart. In contrast to the latter, permissioned protocols do not support public systems 
in which anyone can participate without a specific identify. Hence, permissioned blockchain protocols target a different trust model and thus are out of the scope of our paper.

\textbf{Hybrid-consistency replication.} The dichotomy between weak and strong consistency is well studied 
in the context of traditional geo-replicated systems~\cite{Brewer2000}.
It is well established that one needs to forfeit strong consistency
to obtain the high availability, low latency, partition tolerance and high scalability that geo-replicated systems demand~\cite{Brewer2000}.
It is also known that many geo-distributed applications do 
not require strong consistency for \emph{every} operation~\cite{Lloyd:2011:DSE:2043556.2043593}
and that many such applications are dominated by operations that are correct  even if executed over a weakly-consistent view.
This observation has motivated the advent of geo-replicated systems 
supporting hybrid (or mixed) consistency models~\cite{milano2018mixt, Li:2012:MGS:2387880.2387906}. 
To the best of our knowledge, \name is the first to introduce hybrid consistency models to permissionless environments. 

\section{Conclusions and future work}
\label{sec:conclusion}



This paper proposes \name, which extends standard permissionless blockchains with a fast path that delivers 
\emph{partially ordered promises of commitment}.
This fast path supports cryptocurrency transactions and only takes a small fraction of the original commit latency, while
providing consistency guarantees that are \emph{strong enough} to ensure correct cryptocurrencies.
Since today's general-purpose blockchains also support smart contract transactions, which typically have (strong) sequential consistency needs,
\name implements a \emph{hybrid} consistency model that also supports strongly-consistent applications.
To the best of our knowledge, \name is the first system to bring together fast partially ordered transactions with totally ordered, consensus-based transactions in a permissionless setting.


Our evaluation conducted in a realistic geo-distributed environment with 500 processes shows that the average latency to promise a transaction is an order of magnitude faster than consensus-based commit. 
Furthermore, our empirical evaluation of fragmentation attacks show that, even considering very favourable conditions for 
the attacker, an attacker cannot achieve any tangible benefits in exploring this kind of attack.
A formal analysis of fragmentation attacks is out of the scope of this paper, and left for future work.

Overall, we believe that our approach of bringing fast transactions to permissionless blockchains as an extension to existing blockhains, rather than proposing a new system from scratch, is a step in the direction of bringing these results closer to adoption by \emph{de facto} blockchain systems such as Ethereum or Bitcoin, thus benefiting both the academic and industry communities.
Our work unveils new research avenues. Although this paper focuses on cryptocurrencies as the obvious application to benefit from the promise fast path of \name, our proposal can also provide important benefits to smart contracts which have (a subset of) transactions with weaker consistency needs. As an example, smart contracts that employ the ERC20 Token Standard~\cite{bitcoinwikiERC20} to transfer some asset may have weaker consistency needs. However, providing a hybrid consistency model to smart contract programs 
requires a careful integration of this model into smart contract execution runtimes, 
as well as providing programmers with the adequate abstractions to help them build smart contract methods that can safely run with weaker guarantees.

\section*{Acknowledgements}

This work was supported by national funds through FCT, Fundação para a Ciência e a Tecnologia, with references UIDB/50021/2020 and SFRH/BD/130017/2017, and the European Union's Horizon 2020 -  The EU Framework Programme for Research and Innovation under grant 952226. We thank Miguel Correia, as well as the anonymous reviewers, for their insightful comments on earlier versions of this paper.

\bibliographystyle{ACM-Reference-Format}
\bibliography{bibliography}

\end{document}